\documentclass[a4paper,english,thm-restate,cleveref,colorlinks=true, linkcolor=blue!75!green, urlcolor=blue!75!green, citecolor=blue!75!green,autoref,nameinlink]{lipics-v2021}

\usepackage{dsfont}
\usepackage[inline]{enumitem}
\usepackage{subcaption}
\DeclareTextFontCommand{\emph}{\color{teal}\em}

\newcommand{\N}{\ensuremath{\mathds{N}}}

\newcommand{\bt}{\ensuremath{\mathrm{bt}}}
\newcommand{\wbt}{\ensuremath{\mathrm{wbt}}}
\newcommand{\s}{\ensuremath{\mathcal{S}}}

\crefname{claim}{Claim}{Claims}

\graphicspath{{figures/}}

\newcommand{\prooflink}[1]{\hypersetup{linkcolor=lipicsBulletGray}\hyperref[#1]{$\blacktriangledown$}}
\newcommand{\statlink}[1]{\hypersetup{linkcolor=lipicsYellow}\hyperref[#1]{$\blacktriangle$}}
\newcommand{\mylink}{}
\newcommand{\cprooflink}[1]{\hypersetup{linkcolor=lipicsBulletGray}\hyperref[#1]{$\triangledown$}}
\newcommand{\cstatlink}[1]{\hypersetup{linkcolor=lipicsYellow}\hyperref[#1]{$\vartriangle$}}
\newcommand{\cmylink}{}

\newif\iflinkplaced
\newcommand{\linkhere}{%
  \global\linkplacedtrue
  \unskip\nobreak\hfill\mylink
}

\theoremstyle{claimstyle}

\newcommand{\supplementlink}{\url{https://doi.org/10.5281/zenodo.21622661}}

\title{Weighted Book Thickness}

\author{Henry Förster}{John Cabot University, Rome, Italy $\cdot$ Technical University of Munich, Heilbronn, Germany}{henry.foerster@johncabot.edu}{https://orcid.org/0000-0002-1441-4189}{}

\author{Michael Hoffmann}{Department of Computer Science, ETH Z\"urich, Switzerland \and\url{https://people.inf.ethz.ch/hoffmann/}}{hoffmann@inf.ethz.ch}{https://orcid.org/0000-0001-5307-7106}{}

\author{Stephen Kobourov}{Technical University of Munich, Heilbronn, Germany}{stephen.kobourov@tum.de}{https://orcid.org/0000-0002-0477-2724}{}

\author{Maria Eleni Pavlidi}{University of Ioannina, Greece}{m.e.pavlidi@uoi.gr}{https://orcid.org/0009-0009-4500-0112}{}

\author{Alexandra Weinberger}{FernUniversit\"at in Hagen, Germany}{alexandra.weinberger@fernuni-hagen.de}{https://orcid.org/0000-0001-8553-6661}{}

\author{Johannes Zink}{Technische Universität München, Germany \and \url{https://www.cs.cit.tum.de/algo/staff/johannes-zink/}}{johannes.zink@tum.de}{https://orcid.org/0000-0002-7398-718X}{}

\authorrunning{H. Förster, M. Hoffmann, S. Kobourov, M. E. Pavlidi, A. Weinberger, J. Zink} 

\supplement{Source codes for our computer-aided proofs of \cref{thm:max-nine-vtcs,thm:2tree} are publicly available at \supplementlink.}

\Copyright{Henry Förster, Michael Hoffmann, Stephen Kobourov, Maria Eleni Pavlidi, Alexandra Weinberger, Johannes Zink}

\acknowledgements{This work was initiated at the SWGD'24 workshop. Maria Eleni Pavlidi was supported by HFRI Grant No: 26320.}



\nolinenumbers
\hideLIPIcs

\begin{document}
\maketitle

\begin{abstract}
We introduce and study the weighted book thickness of graphs. A \emph{$k$-page book embedding} of a graph~$G=(V,E)$ is defined by a spanning cycle~$C$ for~$V$ (which does not need to be part of~$G$) and a partition~$E=\bigcup_{i=1}^{k}E_i$ such that~$E\cap C\subseteq E_1$ and each graph~$G_i=(V,E_i\cup C)$, for~$1\le i\le k$, is outerplane with outer cycle~$C$. If~$e\in E_i$, we say that~$e$ appears on Page~$i$. The classical \emph{book thickness} of a graph~$G$ is the minimum~$k$ such that there exists a~$k$-page book embedding of~$G$, that is, the minimum (over all book embeddings of~$G$) achievable \emph{maximum} page an edge appears on.
In contrast, the weighted book thickness is the minimum achievable \emph{average} page an edge appears on.
The embeddings that realize weighted book thickness can differ from those that realize (classical) book thickness. We show that, although every planar graph on at most nine vertices admits a $2$-page book embedding realizing its weighted book thickness, already for ten vertices, there is a planar graph for which every realization of its weighted book thickness needs more pages than its book thickness. We prove that there even exists a $2$-tree whose weighted book thickness cannot be realized on two pages. On the positive side, we show that for every graph of pathwidth at most two, the weighted book thickness can always be realized by a $2$-page book embedding
and such an embedding can be found in linear time.
Moreover, we prove that it is NP-complete to decide if the weighted book thickness is at most~$k$, for some given integer~$k$.
\end{abstract}

\section{Introduction}

Planar graphs have many useful properties. Of course, not all graphs are planar. But  
there exist various approaches to reduce a nonplanar graph to one or multiple planar graphs and, more generally, to measure its distance to planarity.
One such approach
uses the notion of \emph{graph thickness},
that is, the minimum number of planar graphs into which the edges of a graph can be decomposed.
Initiated by Tutte~\cite{tutte1963thickness} and Kainen~\cite{kainen1973thickness},
there has been a great deal of work on graph thickness and several refinements, such
as \emph{geometric thickness}~\cite{DBLP:conf/gd/DillencourtEH98,dillencourt2004geometric,DBLP:journals/comgeo/Duncan11} and \emph{book thickness}~\cite{bernhart1979book}.
In the geometric setting, the graphs of the decomposition
must admit plane straight-line drawings where the same vertices are arranged at the same positions in all drawings.
The book thickness setting additionally requires the vertices to be placed in convex position.\footnote{The notion of vertices arranged on a line, which we describe next, is equivalent to convex positions.}
Relevant results include the asymptotic non-equivalence of graph thickness and geometric thickness~\cite{eppstein2002separating} as well the asymptotic non-equivalence of
geometric thickness and book thickness~\cite{eppstein2001separating}.
See also two surveys on graph thickness~\cite{maekinen2012survey,DBLP:journals/gc/MutzelOS98}.

In a \emph{book embedding} (also known as a \emph{stack layout}) of a graph, all vertices are aligned on a horizontal line called \emph{spine} and the edges are partitioned into subsets.
Each edge is then drawn as a semicircle that is contained entirely within one  half-plane defined by the spine, called \emph{page}, such that two edges in the same page are not allowed to cross. Note that pages may coincide if edges are colored by page and monochromatic crossings are avoided. The book thickness of such an embedding is equal to the number of its pages and the book thickness of a graph is the minimum book thickness of all its book embeddings.

A general upper bound on the book thickness of an $n$-vertex graph is $\lceil n/2 \rceil$~\cite{bernhart1979book}.
Yannakakis has shown in the 1980s that any planar graph has book thickness at most~4~\cite{DBLP:conf/stoc/Yannakakis86,DBLP:journals/jcss/Yannakakis89}, which has been shown to be tight recently~\cite{DBLP:journals/jocg/KaufmannBKPRU20,DBLP:journals/jctb/Yannakakis20}.
Any 2-tree has book thickness at most~2~\cite{DBLP:journals/algorithmica/GiacomoDLW06,DBLP:conf/cocoon/RengarajanM95}, while planar 3-trees have book thickness at most~3~\cite{DBLP:conf/focs/Heath84}.
More general, graphs of treewidth~$k$ have book thickness at most~$k+1$~\cite{DBLP:conf/gd/DujmovicW05,DBLP:journals/dcg/DujmovicW07,DBLP:journals/dam/GanleyH01}.
Determining whether the book thickness of a graph is at most~$k$ is NP-complete already for~$k = 2$ by reduction from Hamiltonicity on planar graphs~\cite{CLR87}.
Related types of embeddings have also been investigated~\cite{DBLP:conf/gd/BekosFKKKR22,DBLP:conf/cccg/BekosKPR23,DBLP:conf/wads/GiacomoDFU025,DBLP:journals/comgeo/GiacomoDLW05,DBLP:journals/siamcomp/HeathR92,DBLP:conf/gd/Pupyrev17}.

Motivation for studying geometric thickness and book thickness can be found in VLSI design~\cite{CLR87}, in which the vertices represent components of a circuit and the wires represent connections between them. The wires are assigned to the minimum number of layers such that on each layer the corresponding part of the circuit is crossing-free. In addition, the models could be used to model the routing of pipes and cables on different depths underground in civil engineering applications. Next to these modeling applications, book embeddings also have several applications in graph drawing, where two of the standard visualization styles for graphs, arc diagrams~\cite{chktw-adfdht-15,chktw-adfdht-18,thesisFoerster,DBLP:conf/infovis/Wattenberg02} and circular layouts~\cite{DBLP:conf/wg/BaurB04,DBLP:phd/dnb/Gronemann15}, can be constructed using book embeddings.

In this paper, we introduce the \emph{weighted book thickness} of graphs.
While classical book thickness is the minimum achievable number of pages in a book embedding, weighted book thickness is the minimum achievable average page number assigned to an edge in a book embedding.
Another way to interpret weighted book thickness is to consider different costs for placing edges on different pages; e.g.: 1€ per edge on Page~1, 2€ per edge on Page~2, and so on.
For example, in civil engineering, the cost of routing pipes  might be higher the deeper underground trenches need to be dug~\cite{10.1680/iicep.1972.5576,EGGIMANN2015218} while, in VLSI design, expensive semiconductor material might be used  only where it is actually needed to  route connections~\cite{DBLP:conf/iccad/Stow0SL17}.

For vertex coloring, the analogous concept
of taking the average instead of the maximum
has been studied under the name \emph{sum coloring}~\cite{DBLP:conf/approx/GiaroJKM02,DBLP:conf/random/HalldorssonKS01,DBLP:journals/arscom/Kubicka05,DBLP:conf/acm/KubickaS89,DBLP:journals/orl/Marx05,DBLP:journals/tcad/Supowit87},
that is, for a graph $G = (V, E)$, find
a function $c: V \rightarrow \mathbb{N}^+$
such that for each $\{u, v\} \in E \colon c(u) \ne c(v)$
and $\sum_{v \in V} c(v)$ is minimized.
In general, this problem is NP-hard
but it becomes polynomial-time solvable
for trees, pseudotrees, and outerplanar graphs
\cite{DBLP:journals/arscom/Kubicka05,DBLP:conf/acm/KubickaS89}.
Obtaining a sum coloring sometimes requires
more distinct colors than the chromatic number;
the number of additional colors
may be arbitrarily large~\cite{erdoes-sum-coloring}.
It remains NP-hard on interval graphs~\cite{DBLP:journals/orl/Marx05},
but there are approximations~\cite{DBLP:conf/approx/GiaroJKM02,DBLP:conf/random/HalldorssonKS01}.

\subparagraph*{Our contribution.}
We first give a simple upper bound on the
weighted book thickness (\cref{p:upper})
and a lower bound for planar 3-trees (\cref{p:lower}).
Then, we evaluate the computational complexity with a similar result as for the (classical) book thickness --
determining the weighted book thickness of a graph is NP-complete (\cref{thm:np-hard}).
Our main contribution is to investigate whether the book thickness and weighted book thickness can be separated. Indeed, we show that  the embeddings realizing weighted book thickness can differ from those realizing (classical) book thickness.
We show that, although every planar graph on at most nine vertices admits a $2$-page book embedding realizing its weighted book thickness (\cref{thm:max-nine-vtcs}),
already for ten vertices, there is a planar graph of pathwidth~3 for which every realization of its weighted book thickness needs more pages than its book thickness (\cref{thm:3-path-distinct}).
Moreover, we prove that there is a $2$-tree whose weighted book thickness cannot be realized on two pages (\cref{thm:2tree}).
As a complementing positive result, we show that, for every graph of pathwidth at most~2,
the weighted book thickness can always be realized by a $2$-page book embedding
and such an embedding can be found in linear time (\cref{thm:pathwidth2}).

\section{Preliminaries}

\subparagraph{(Weighted) Book Embeddings.} 
A \emph{$k$-page book embedding} of a graph~$G=(V,E)$ is a pair~$(p,C)$, where~$p:E\to\{1,\ldots,k\}$ is a partition of the edges of~$G$ into~$k$ \emph{pages}, and~$C=(V,E_C)$ is a  cycle (whose edges are not necessarily in~$G$; we call this order of vertices also \emph{order on the spine})
such that for each color class~$E_i=p^{-1}(i)\subseteq E$, for~$i\in\{1,\ldots,k\}$, the graph~$G_i=(V,E_i\cup E_C)$ is outerplanar with outer cycle~$C$.
The \emph{book thickness}~\cite{bernhart1979book,Oll73}  $\bt(G)$ of~$G$ (also called \emph{stack number}) is the minimum~$k\ge 1$ for which there exists a~$k$-page book embedding of~$G$.
The \emph{weight} of a book embedding~$(p,C)$ of a graph~$G=(V,E)$ is defined by~$w(p,C)=\sum_{e\in E}p(e)$.
The \emph{weighted book thickness on $k$ pages} of~$G$ is defined as
\begin{equation*}
    \wbt_k(G)= \begin{cases}
        \infty ~~~~~~~~~~~~~~\,\text{, if } \mathrm{bt}(G) > k \\
        \min\limits_{(p,C)} \frac{w(p,C)}{|E|} ~~~\text{, otherwise}
    \end{cases}
\end{equation*}\stepcounter{linenumber}where the minimum is taken over each $k$-page book embedding~$(p, C)$ of~$G$.
Moreover, we define the \emph{weighted book thickness} of~$G$ as
\begin{equation*}
    \wbt(G)=\min_{k\in\N}\wbt_k(G).
\end{equation*}
While book thickness considers the total number of pages only,
the weighted version provides a model
that measures the average page number. As a result, it requires a more careful analysis of the assignment of individual edges to pages; see e.g., \cref{fig:ten}.

The graphs admitting $1$-page book embeddings are exactly the outerplanar graphs~\cite{bernhart1979book}. This family of graphs can also be defined by the following two forbidden substructures:
\begin{theorem}[{\cite[Theorem 1]{Chartrand1967PlanarPG}}]
    \label{thm:outerplanar}
    A graph $G$ is outerplanar if and only if it contains no subdivision of $K_{2,3}$ or $K_4$.
\end{theorem}

\begin{figure}[t]
\centering
\includegraphics[page=1]{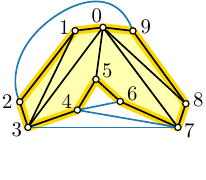}\hspace{48pt}\includegraphics[page=2]{figures/example.pdf}\\
\includegraphics[page=2]{ten}\hspace{48pt}\includegraphics[page=1]{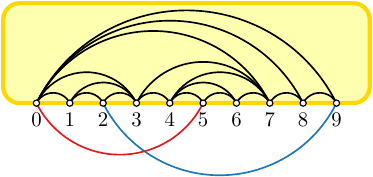}
\caption{The planar graph~$G$ from \cref{thm:3-path-distinct}, for which~$\wbt(G)=\wbt_3(G)<\wbt_2(G)$. Edges assigned to Pages $1$, $2$ and $3$ are colored black, blue and red, respectively. The Hamiltonian cycle $C$ inducing the spine orders is colored yellow.
\textbf{(top)}~A plane embedding of $G$ with a planar Hamiltonian cycle $C$ such that all but four edges are contained within $C$. This gives rise to a $2$-page book embedding with weighted book thickness $23/19$.
\textbf{(bottom)}~An embedding of $G$ with a crossing such that cycle $C$ is a planar Hamiltonian cycle for all edges except for edge $(0,5)$. This embedding and $C$ induce a $3$-page book embedding with weighted book thickness $22/19$.  
}
\label{fig:ten}
\end{figure}

\subparagraph{Subhamiltonian Graphs.}
A graph is \emph{subhamiltonian} if it is a subgraph of a planar Hamiltonian graph, that is, a planar graph that contains a Hamiltonian cycle. It is well known that all planar graphs on up to ten vertices are subhamiltonian, and on~$11$ vertices there is one exception only: the so-called Goldner-Harary graph~\cite{GH75}.
Moreover, all graphs that admit a~$2$-page book embedding are subhamiltonian. Every such graph~$G$ is planar, as we can draw the two edge sets~$E_1$ and~$E_2$ on different sides of the cycle~$C$ to obtain a plane embedding. Adding the edges of~$C$---if not present in~$G$ already---we obtain a Hamiltonian plane graph.%

A \emph{planar subhamiltonian cycle} of a plane graph~$G$ is a cycle through all vertices of~$G$ that may use edges and non-edges of~$G$, such that the non-edges may be added to the embedding of~$G$ without crossings. Observe that such a cycle $C$ defines a $2$-page book embedding where, w.l.o.g., the edges within $C$ are assigned to Page~$1$ and the cycles outside of $C$ are assigned to Page~$2$; see \cref{fig:ten} (top) for an example. 

This notion of a certificate cycle $C$ for $2$-page embeddability can also be extended to a \emph{certificate cycle for $3$-page embeddability} as follows: $C$ must be a planar subhamiltonian cycle for the subgraph induced by the first two pages and the third page includes the edges crossing~$C$ and the edges crossing an edge from one of the first two pages.
Moreover, each pair of such edges on the third page must not have their endpoints interleaving in a walk along~$C$ (this would correspond to a crossing in the book embedding).
For a small number of edges on Page~$3$, the validity is straight-forward to check.
For an example, refer to \cref{fig:ten} (bottom).

\subparagraph{Tree and Path Decompositions.} A \emph{tree decomposition} of a graph $G=(V,E)$ is a $2$-tuple $(T,\mathcal{B})$ where
\begin{itemize}
    \item $T=(V_T,E_T)$ is a tree with $V_T=\mathcal{B}$,
    \item each $B\in \mathcal{B}$ is a subset of $V$, called a \emph{bag}, so that $\bigcup_{B \in \mathcal{B}}B=V$,
    \item for each edge $uv \in E$, there is a bag $B \in \mathcal{B}$ so that $u\in B$ and $v\in B$,
    \item for each vertex $v \in V$, the elements of $\mathcal{B}$ containing $v$ induce a connected subgraph of $T$.
\end{itemize}
Moreover, a \emph{path decomposition} of $G$ is a tree decomposition $(P,\mathcal{B})$ where $P$ is a path. 

The \emph{width} of tree decomposition $(T,\mathcal{B})$ is equal to the maximum cardinality of any element of $\mathcal{B}$ minus $1$, i.e., $\mathrm{width}(T,\mathcal{B})=\max_{B\in \mathcal{B}}\{|B|-1\}$. The \emph{treewidth} $\mathrm{tw}(G)$ of graph $G$ is the minimum width of any of its tree decompositions. Moreover, the \emph{pathwidth} $\mathrm{pw}(G)$ of a graph $G$ is the minimum width of any of its path decompositions. 
The maximal graphs with treewidth~$k$ are the so-called $k$-trees.
A \emph{$k$-tree} is a graph that can be obtained by
starting with the complete graph $K_{k+1}$
and repeatedly adding a vertex~$v$ of degree exactly~$k$
such that neighbors of $v$ form a clique. 
Observe that every 2-tree is planar.

\section{Upper and Lower Bounds}

We first observe the following upper bound on the weighted book thickness:

\begin{proposition}\label{p:upper}
  For every graph~$G$ with~$\bt(G)\le k$, we have $\wbt(G)\le(k+1)/2$.
\end{proposition}
\begin{proof}
    Consider any $k$-page book embedding~$(p,C)$ of~$G$, and permute the order of the pages in~$p$ so that~$|E_i|$ is non-increasing in~$i$.
    In the worst case, every page has exactly $m/k$ edges and its edge weights sum to $mi/k$, where~$m$ denotes the number of edges in~$G$. This yields:
    \begin{equation*}
        \wbt(G) \le \frac{1}{m} \sum_{i=1}^{k}\frac{mi}{k} = \sum_{i=1}^{k}\frac{i}{k} = \frac{k+1}{2}. \qedhere
    \end{equation*}
\end{proof}

\noindent Planar~$3$-trees have book thickness at most three~\cite{DBLP:conf/focs/Heath84}, so by \cref{p:upper} their weighted book thickness is at most two. Below we give a lower bound of~$13/9\approx 1.44$.

\begin{proposition}\label{p:lower}
  There exists an infinite family~$G_i$ of planar~$3$-trees with~$\wbt(G_i)\ge 13/9-o(1)$.
\end{proposition}
\begin{proof}
   Cardinal, Hoffmann, Kusters, T{\'o}th, and Wettstein~\cite{chktw-adfdht-15,chktw-adfdht-18} exhibit a family~$G_i$, for~$i\ge 1$, of planar~$3$-trees such that~$G_i$ has~$n_i=3i+8$ vertices and at least~$(n_i-8)/3=i$ edges have to removed from~$G_i$ in order to obtain a subhamiltonian planar graph. Consider any book embedding of~$G_i$. 
   Note that~$G_i$ has~$3n_i-6=9i+18$ edges, of which at most~$2n_i-3=6i+13$ are on Page~$1$ because every page induces an outerplane subgraph of~$G_i$. Of the remaining at least~$3i+5$ edges, at least~$i$ are on Page~$3$ (or higher). Therefore, we can bound
   \[
   \wbt(G_i)\ge \frac{6i+13+2(2i+5)+3i}{9i+18}=\frac{13i+23}{9i+18}=\frac{13}{9}-\frac{1}{3(i+2)}\,.\qedhere
   \]
\end{proof}

\section{Computational Complexity}
It is a natural question to ask for the computational complexity of determining the weighted book thickness of a graph.
Similar to other versions of book thickness, it is NP-complete.

\begin{theorem}
\label{thm:np-hard}
    Given a graph $G$, it is NP-complete to 
    decide if the weighted book thickness~$\wbt(G)$ is at most~$k$, for some given integer~$k$.
    This holds even if $G$ is a planar triangulation.
\end{theorem}
\begin{proof}
    The problem is in NP:  given an order of the vertices on the spine and assignment of the edges to pages,
    the weighted book thickness can be verified in polynomial time.

    To show NP-hardness, we reduce from the NP-hard problem Hamiltonicity on triangulations~\cite{wigderson1982hamiltonian}.
    Namely, we show that a planar triangulation~$G$ with $n$ vertices
    is Hamiltonian if and only if its minimum weighted book embedding has a weight of $4n-9$.
    Recall that a graph admits a $2$-page book embedding if and only if it is subhamiltonian.
    
    First assume that $G$ is Hamiltonian. Then, it admits a 2-page book embedding.
    Since $G$ is a triangulation, it has $3n-6$ edges.
    Since each of the two pages induces an outerplane graph, they share $n$ edges, $n-1$ that follow the order on the spine and $1$ between the first and the last vertex.
    These edges can be assigned to Page~$1$.
    The remaining $2n-6$ edges must be assigned to Pages~$1$ and~$2$, namely, 
    $n-3$ edges each. Thus, we have a weight of $4n-9$.

    Second, assume that $G$ is not Hamiltonian. Then, there is no $2$-page book embedding and we must assign at least one edge to Page~$3$.
    Since Page~$1$ contains at most $2n-3$ edges and Page~$2$ at most $n-4$ edges,
    we get a weight of at least $2n-3 + 2 \cdot (n-4) + 3 = 4n-8>4n-9$.
\end{proof}

\section{Planar Graphs of Pathwidth~3}\label{sec:pathwidth3}
Any book embedding of a graph~$G$ necessarily uses at least~$\bt(G)$ pages. As increasing the number of pages also increases the contribution of the edges on these pages in the weight of the embedding, one might think that considering~$\bt(G)$ pages might also be sufficient to obtain an embedding with optimum weighted book thickness, that is, such that~$\wbt(G)=\wbt_{\bt(G)}(G)$. 
In this section, we show that this is not the case in general. 
We show that there is a graph on ten vertices whose weighted book thickness cannot be realized on two pages, even though it has book thickness two.

\begin{theorem}
\label{thm:3-path-distinct}
  There is a planar graph~$G$ on~$10$ vertices of pathwidth~$3$ with~$22/19=\wbt(G)<\wbt_2(G)=23/19$.
\end{theorem}

\begin{proof}
    Let~$G$ be the planar graph with~$10$ vertices and~$19$ edges depicted in \cref{fig:ten}.
    It has pathwidth at most~$3$ because it
    admits a path decomposition $(P,\mathcal{B})$ of width~$3$ where $P=(B_1,B_2,B_3,B_4,B_5,B_6,B_7)$ and $B_1=\{1,2,3,9\}$, $B_2=\{0,1,3,9\}$, $B_3=\{0,3,8,9\}$, $B_4=\{0,3,7,8\}$, $B_5=\{0,3,4,7\}$, $B_6=\{0,4,5,7\}$ and $B_7=\{4,5,6,7\}$.
    Moreover, $G$ has pathwidth at least~$3$
    because it contains $K_4$ as a minor,
    which has pathwidth~$3$, and pathwidth is closed under taking minors.
    Hence, the pathwidth of~$G$ is exactly~$3$.
    
    To see that~$\wbt(G)\le 22/19$, consider the drawing shown in \cref{fig:ten}~(bottom):
    The black edges form a maximal outerplane graph, which we put onto Page~$1$, and the remaining two edges~$\{2,9\}$ and~$\{0,5\}$ are put onto Pages~$2$ and~$3$, respectively. The weight of the resulting~$3$-page book embedding is~$17\cdot 1+1\cdot 2+1\cdot 3=22$. Among all book embeddings of~$G$ that use three or more pages, this is best possible because Page~$2$ and Page~$3$ each contain the minimum of one edge. Then, $\wbt(G)=22/19$ follows from~$\wbt_2(G)\ge 23/19$, which we will show next.
    
    Consider a~$2$-page book embedding~$(p,C)$ of~$G$. We claim that~$p$ assigns at least four edges to Page~$2$, which implies~$\wbt_2(G)\ge (15 + 4 \cdot 2)/19 = 23/19$.
    Observe that~$G$ is~$3$-connected and, therefore, by Whitney's Theorem~\cite{Whitney19323connected} it has a unique plane embedding~$\Gamma$ (on the sphere), up to mirroring, shown in \cref{fig:ten}~(top).
    Every~$2$-page book embedding of~$G$ must respect~$\Gamma$, so this also holds for~$(p,C)$. In particular, all edges of~$C$ are either in~$G$ or their union with~$\Gamma$ is plane. The edges on Page~1 form an outerplane graph on~$10$ vertices, which has at most~$2\cdot 10-3=17$ edges
    and at most $10 - 2 = 8$ inner faces.
    If at most three edges are on Page~2, then at least~$16$ edges are on Page~1.
    Thus, in~$(p,C)$ the graph~$G_1$ (i.e., the graph of Page~1) is a maximal outerplane graph minus at most one edge.
    In particular, at most one edge of~$C$ is not an edge of~$G$.
    So at most one face of~$\Gamma$ is split by an edge of~$C$ and thus contributes to both sides of~$C$. All other faces lie on either side of~$C$, and all faces on the same side of~$C$ form a subtree of the dual graph~$G^*$ of~$\Gamma$.
    Further, as~$C$ is a subhamiltonian cycle of~$G$, every vertex is incident to (at least one face of) both sides of~$C$.
    When speak of faces of~$G_1$,
    we refer to the inner faces derived from~$\Gamma$
    when considering the edges of~$G_1$.
    
    If~$C\subset G$, then $G_1$ has seven or eight faces.
    All of these faces except maybe one are triangles.
    Moreover, they are also inner faces of~$\Gamma$, and they form a subtree of~$G^*$.
    Observe that $\Gamma$ has seven triangular faces, split into three groups in~$G^*$, separated by either a quadrilateral or a pentagon.
    Thus, every subgraph of~$G^*$ that contains six or more triangles also contains at least two quadrilaterals or one pentagon, neither of which is possible for the faces of~$G_1$.
    
    Otherwise, the cycle~$C$ splits one quadrilateral or pentagonal face of~$\Gamma$ into two;
    as the two resulting faces are on different sides of~$C$, at most one of them is also a face of~$G_1$.
    Furthermore, $G_1$ has eight faces, which are all triangles.
    Thus, $G_1$ comprises at most two of the three groups of triangles of~$G$, that is, at most five triangles of~$G$, plus at most one more triangle resulting from the face split, which is less than eight.
    
    Therefore, every~$2$-page book embedding~$(p,C)$ of~$G$ has at least four edges on Page~2, for a weight of at least~$15+4\cdot 2=23$. 
\end{proof}

In fact, it turns out that our separation example in \cref{thm:3-path-distinct} is the smallest-order example -- independent of pathwidth, as we show in the following theorem:

\begin{theorem}
    \label{thm:max-nine-vtcs}
    For every planar graph~$G$ on~$n\le 9$ vertices we have~$\wbt(G)=\wbt_2(G)$, that is, there exists a~$2$-page book embedding that realizes~$\wbt(G)$.
\end{theorem}
\begin{proof}
    We prove the statement using brute-force computation; our source code is available at \supplementlink.
    We generate all biconnected planar graphs on~$n\le 9$ vertices. For each such graph~$G=(V,E)$, we~iteratively consider all sets~$F\subset E$ of~$\mu$ edges, for~$\mu=0,1,\ldots$, and test if~$G\setminus F$ is outerplanar. If so, we go over all possible outerplane embeddings~$\Gamma$ of~$G\setminus F$ and test if~$F\cup C$ is outerplanar, where~$C$ is the outer cycle of~$\Gamma$. If this is the case, we found a~$2$-page embedding~$(p,C)$ of~$G$, where~$p(e)=2$, for~$e\in F$, and~$p(e)=1$, otherwise. Given that we have already considered all possible~$2$-page embeddings with strictly fewer edges on Page~2 before, the embedding~$(p,C)$ minimizes~$\wbt(G)$ among all embeddings using at most two pages (one page corresponds to~$\mu=0$). As~$G$ is subhamiltonian planar, the procedure terminates for some~$\mu\le|E|/2$. 
    
    It remains to consider possible improvements using embeddings with three or more pages.
    In order to improve on the best~$2$-page embedding found, for any~$x$ edge(s) moved to Page~$3$ (or higher), we need to move at least~$x+1$ edges from Page~$2$ to Page~$1$. So for any single edge~$e$ to be put on Page~$3$, we test if~$G\setminus e$ admits a~$2$-page embedding with at most~$\mu-3\ge 1$ edge(s) on Page~$2$. No such embedding is found for any planar graph on~$n\le 9$ vertices. For two edges~$e,f$ to be put on Page~$3$, we have to find a~$2$-page embedding of~$G\setminus\{e,f\}$ with at most~$\mu-5\ge 2$ edges on Page~$2$.
    As it turns out, every planar graph on~$n\le 9$ vertices admits a~$2$-page embedding with at most six edges on Page~$2$. Therefore, there is no need to test for embeddings with two or more edges on Page~$3$ or higher.  
\end{proof}

\section{2-Trees}
\label{sec:2trees}

\begin{figure}
    \centering
    \begin{subfigure}{0.3\textwidth}
    \includegraphics[width=\textwidth,page=4]{figures/2tree.pdf}
    \subcaption{}
    \label{fig:2tree:construction:1}
    \end{subfigure}
    \hfill
    \begin{subfigure}{0.65\textwidth}
    \includegraphics[width=\textwidth,page=1]{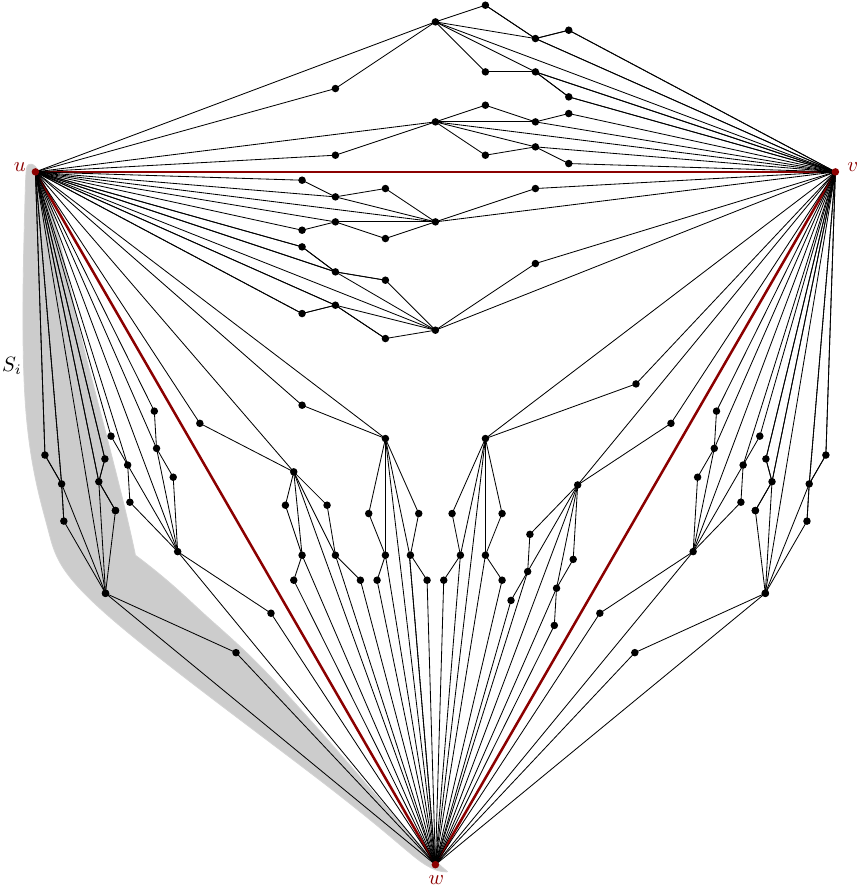}
    \subcaption{}
    \label{fig:2tree:construction:2}
    \end{subfigure}
    \caption{(a)~A sunflower graph $S$ and (b)~the 2-tree $G$ used in our proof of \cref{thm:2tree}.}
    \label{fig:2tree:construction}
\end{figure}

After considering graphs with pathwidth~3, we continue with graphs of treewidth~2.
Recall that the edge-maximal graphs with treewidth~2 are 2-trees,
and all 2-trees are planar.
The remainder of this section is devoted to proving the following theorem:

\begin{theorem}
\label{thm:2tree}
    There is a 2-tree $G$, with $|V|=99$ and $|E|=195$ such that $220/195  = \wbt_3(G) < \wbt_2(G)=221/195$.
\end{theorem}

\subparagraph*{Construction of Graph $G$.}~Graph $G$ consists of several subgraphs, each isomorphic to the following \emph{sunflower graph} $S$. The graph $S$ consists of $10$ vertices and $16$ edges composed of two $2$-connected components defined as follows; see~\cref{fig:2tree:construction:1}:   
\begin{itemize}
    \item An $8$-cycle $(s_1,s_2,s_3,s_4,s_5,s_6,s_7,s_8)$ with five chords $s_1s_3$, $s_1s_5$, $s_1s_7$, $s_3s_5$, and $s_7s_5$.
    \item A $3$-cycle $(s_5,s_9,s_{10})$. Note that $s_5$ is shared with the previous $2$-connected component.
\end{itemize}
Based on $S$ we construct $G$ as follows. Let $\Delta =(u,v,w)$ be a triangle. For each edge $xy$ ($x,y\in\{u,v,w\}$ with $x\neq y$) of $\Delta$, we attach four copies $S_i^{xy}, i=1,2,3,4$,
by identifying the vertices $x$ and $y$ with the vertices $s_1$ and $s_{10}$ of each copy, as described below; see \cref{fig:2tree:construction:2}:
\begin{itemize}
    \item For $S_1$ and $S_2$ identify their copies of $s_1$ with $x$ and their copies of $s_{10}$ with $y$.
    \item For $S_3$ and $S_4$ identify their copies of $s_1$ with $y$ and their copies of $s_{10}$ with $x$.
\end{itemize}
This completes the construction of graph $G$. Observe that $G$ has $n=3+3\cdot4\cdot8=99$ vertices and $m=3+3\cdot4\cdot16=195$ edges. We show the following:

\begin{restatable}{rlemma}{twotreetreewidth}
    \label{prop:2tree-treewidth}
    Graph $G$ has treewidth~$2$.
\end{restatable}

\begin{proof}
    We show that $G$ has indeed treewidth $2$. To this end, we first describe the tree decomposition of $S$; see \cref{fig:2tree:treewidth:1}:
\begin{itemize}
    \item The root of the tree decomposition $\mathcal{T}_S$ of $S$ is the bag $B_S=\{s_1, s_5, s_{10}\}$. Observe that this bag contains both endpoints of the edges $s_1s_5$ and $s_1s_{10}$. 
    \item The bag $B_S$ has three children: $\{s_1,s_3,s_5\}$, $\{s_5,s_9,s_{10}\}$ and $\{s_1,s_7,s_5\}$. These bags contain both endpoints of the edges $s_1s_3$, $s_3s_5$; $s_5s_9$, $s_9s_{10}$; and $s_1s_7$, $s_7s_5$, respectively.
    \item The two children of $\{s_1,s_3,s_5\}$ are the bags $\{s_1,s_2,s_3\}$ and $\{s_3,s_4,s_5\}$ containing both endpoints of the edges $s_1s_2$, $s_2s_3$ and $s_3s_4$, $s_4s_5$, respectively.
    \item Similarly, the two children of $\{s_1,s_7,s_5\}$ are the bags $\{s_1,s_8,s_7\}$ and $\{s_7,s_6,s_5\}$ containing both endpoints of the edges $s_1s_8$, $s_8s_7$ and $s_7s_6$, $s_6s_5$, respectively.
\end{itemize}
It follows that for every edge of $S$, both endpoints appear together in at least one bag of $\mathcal{T}_S$. Moreover, it is straightforward to verify that, for each vertex $s_i$ with $i\in \{1,\ldots,10\}$, the set of bags containing $s_i$ induces a connected subtree of $\mathcal{T}_S$. Hence, $\mathcal{T}_S$ is a valid tree decomposition of $S$ of width $2$, rooted at $B_S$.
Based on $\mathcal{T}_S$ we now construct a tree decomposition $\mathcal{T}$ of $G$; see \cref{fig:2tree:treewidth:2}:

\begin{itemize}
    \item The root of $\mathcal{T}$ is a bag $B_r$ containing the three vertices of $\Delta$, i.e., $u$, $v$ and $w$. In particular, it contains both endpoints of the edges $uv$, $uw$ and $vw$.
    \item For each subgraph $S_{i}^{xy}$ isomorphic to $S$ (i.e., $i \in \{1,...,4\}$ and $x,y\in\{u,v,w\}$, $x\neq y$), we obtain the tree decomposition $\mathcal{T}_{S_i^{xy}}$ isomorphic to $\mathcal{T}_S$ as discussed above. It has the bag $B_{S_i^{xy}}$ containing the two vertices $x$ and $y$. We make the bag $B_{S_i^{xy}}$ a child of $B_r$.
\end{itemize}
Since $B_r$ contains the three edges not belonging to any copy of $S$, $\mathcal{T}$ contains all vertices and the two endpoints of each edge in at least one bag. Moreover, each bag has width at most $2$. Finally, the bag $B_r$ guarantees that the subtrees induced by the bags containing each vertex are connected. Namely, two copies ${S_i^{xy}}$ and ${S_j^{ab}}$ of $S$ are vertex-disjoint except for at most two of the vertices $u$, $v$ and $w$, which are all contained both in $B_r$ as well as in $B_{S_i^{xy}}$ and $B_{S_j^{ab}}$. Thus, $\mathcal{T}$ is a valid tree decomposition of $G$ of width $2$.
\end{proof}

\begin{figure}
    \centering
    \begin{subfigure}{0.6\textwidth}
    \includegraphics[width=\textwidth,page=5]{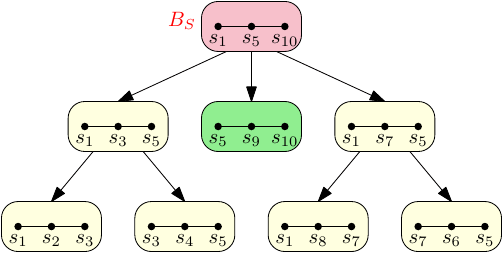}
    \subcaption{}
    \label{fig:2tree:treewidth:1}
    \end{subfigure}
    \hfill
    \begin{subfigure}{0.35\textwidth}
    \includegraphics[width=\textwidth,page=6]{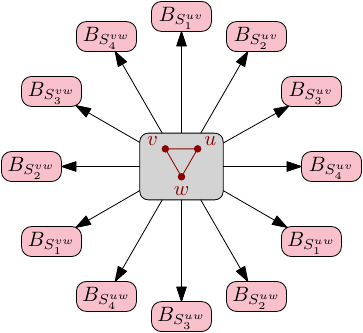}
    \subcaption{}
    \label{fig:2tree:treewidth:2}
    \end{subfigure}
    \caption{Tree decompositions of width $2$ for (a)~the sunflower graph $S$ and (b)~the two-tree $G$.}
    \label{fig:2tree:treewidth}
\end{figure}

\subparagraph*{(Near-)Maximal Outerplanar Subgraphs.} We next consider necessary properties of a large-cardinality edge set $E_1^*\subset E(G)$ that can be assigned to the first page in any book embedding. All our proofs rely on the fact that the induced subgraph $G[E_1^*]$ must be an outerplanar graph~\cite{bernhart1979book} as otherwise it would not admit a $1$-page book embedding, allowing us to argue using \cref{thm:outerplanar}. 

We first consider the restriction of $E_1^*$ to the subgraph $G^{xy}$ induced by the edge $xy$ and the four sunflower subgraph $S_1^{xy}$, $S_2^{xy}$, $S_3^{xy}$ and $S_4^{xy}$. In the following, we will denote the copy of vertex $s_j$ in subgraph $S_i^{xy}$ by $s_{i,j}^{xy}$.

\begin{lemma}
    \label{lem:dense-component}
    Let $G^{xy}$ be the subgraph of $G$ induced by the edge $xy \in \{uv,uw,vw\}$ and the four sunflower subgraph $S_1^{xy}$, $S_2^{xy}$, $S_3^{xy}$ and $S_4^{xy}$. Moreover, let $E_1^{xy}$ be a  subset of edges $E^{xy}$ of $G^{xy}$ that induce an outerplanar subgraph. Then, $|E_1^{xy}|\leq 58 = |E^ {xy}|-7$. Moreover, if $|E_1^{xy}| \geq 56$, the missing edges include the following:
    \begin{enumerate}[label=\textsf{\textbf{\textcolor{lipicsGray}{\upshape{{\arabic*}.}}}},ref=\textsf{\textbf{\textcolor{lipicsGray}{\upshape{{\arabic*}.}}}}]
        \item\label{prop:dense:1} $s_{1,5}^{xy}y, s_{3,5}^{xy}x \notin E_1^{xy}$ (2 missing edges)
        \item\label{prop:dense:2} either $s_{1,5}^{xy}s_{1,9}^{xy} \notin E_1^{xy}$ or $s_{1,9}^{xy}y \notin E_1^{xy}$ (1 missing edge)
        
        \item\label{prop:dense:3} either $s_{3,5}^{xy}s_{3,9}^{xy} \notin E_1^{xy}$ or $s_{3,9}^{xy}x \notin E_1^{xy}$ (1 missing edge)
        
        \item\label{prop:dense:4} $s_{2,5}^{xy}y \notin E_1^{xy}$ or $s_{4,5}^{xy}x \notin E_1^{xy}$ (at least 1 missing edge)

        \item\label{prop:dense:5} either $xy \notin E_1^{xy}$ or both $s_{2,5}^{xy}y \notin E_1^{xy}$ and $s_{4,5}^{xy}x$ (1 missing edge)
        \end{enumerate}
        \vspace{-17pt}
        \begin{enumerate}[label=\textsf{\textbf{\textcolor{lipicsGray}{\upshape{6{\alph*}.}}}},ref=\textsf{\textbf{\textcolor{lipicsGray}{\upshape{6{\alph*}.}}}}]
        \item\label{prop:dense:6a}
        if $s_{2,5}^{xy}y \notin E_1^{xy}$, then also either $s_{2,5}^{xy}s_{2,9}^{xy} \notin E_1^{xy}$ or $s_{2,9}^{xy}y \notin E_1^{xy}$ (1 missing edge)
         \item\label{prop:dense:6b} if $s_{4,5}^{xy}x \notin E_1^{xy}$, then also either $s_{4,5}^{xy}s_{4,9}^{xy} \notin E_1^{xy}$ or $s_{4,9}^{xy}x \notin E_1^{xy}$ (1 missing edge). \linkhere 
        \end{enumerate}
\end{lemma}

\begin{proof}
    We iteratively apply \cref{thm:outerplanar} to identify subdivisions of $K_{2,3}$ for which at least one edge must be removed from the first page to guarantee that the first page induces an outerplanar subgraph of $G^{xy}$. When deciding, which edge of a sunflower subgraph $S$ of $G^{xy}$ to remove from Page~$1$, we make use of the structure of $S$: 
    There is one subdivision $(s_5,s_9,s_{10})$ of edge $s_5s_{10}$ while for $s_1s_5$ there are four edge-disjoint subdivisions: $(s_1,s_2,s_3,s_4,s_5)$, $(s_1,s_3,s_5)$, $(s_1,s_7,s_5)$, $(s_1,s_8,s_7,s_6,s_5)$; see \cref{fig:2tree:dense-component:1}.

\begin{figure}[t]
    \centering
    \begin{subfigure}{\textwidth}
    \includegraphics[width=\textwidth,page=7]{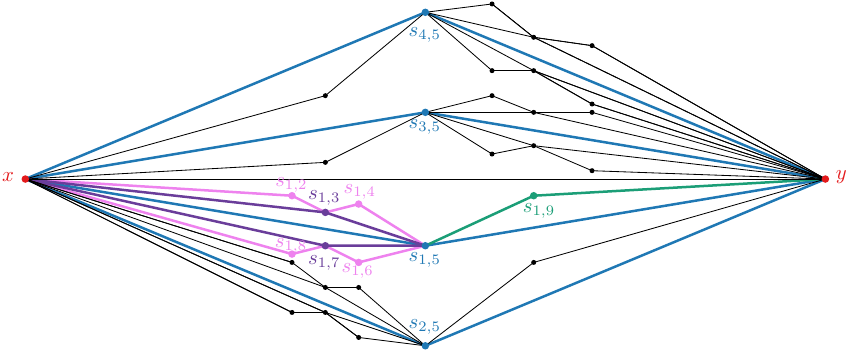}
    \subcaption{}
    \label{fig:2tree:dense-component:1}
    \end{subfigure}
    \hfill
    \begin{subfigure}{\textwidth}
    \includegraphics[width=\textwidth,page=8]{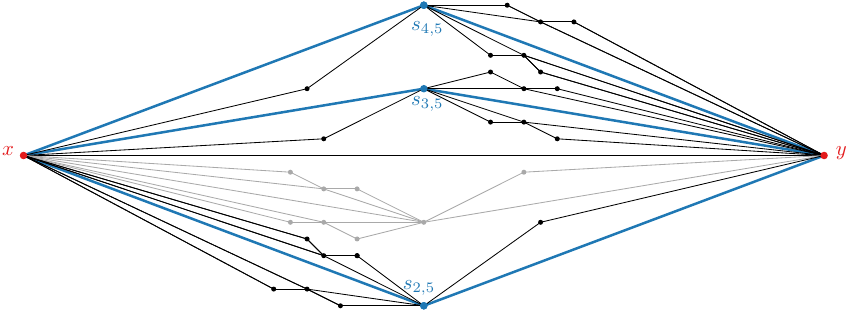}
    \subcaption{}
    \label{fig:2tree:dense-component:2}
    \end{subfigure}
    \caption{Illustrations for the proof of \cref{lem:dense-component}.}
    \label{fig:2tree:dense-component:firstpart}
\end{figure}
    
\begin{figure}[p]
    \begin{subfigure}{\textwidth}
    \includegraphics[width=\textwidth,page=9]{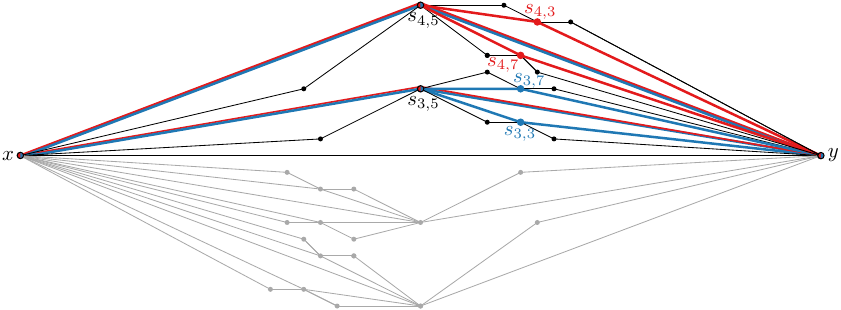}
    \subcaption{}
    \label{fig:2tree:dense-component:3}
    \end{subfigure}

    \begin{subfigure}{\textwidth}
    \includegraphics[width=\textwidth,page=10]{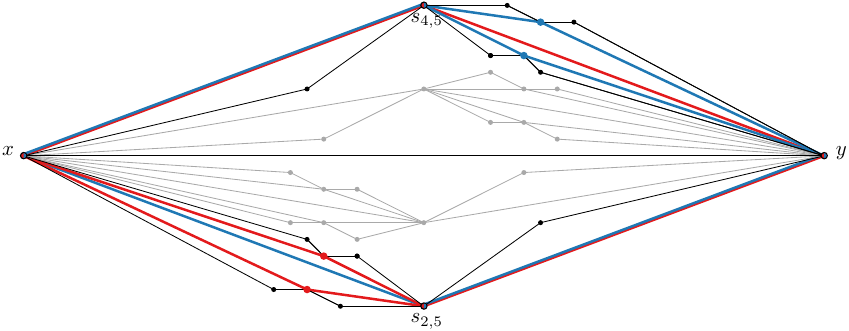}
    \subcaption{}
    \label{fig:2tree:dense-component:4}
    \end{subfigure}

    \begin{subfigure}{\textwidth}
    \includegraphics[width=\textwidth,page=11]{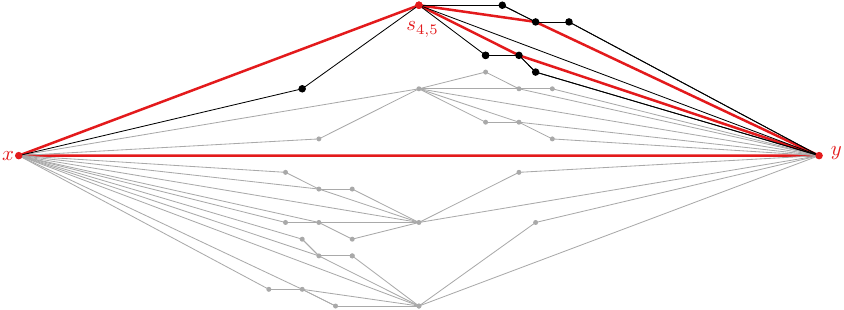}
    \subcaption{}
    \label{fig:2tree:dense-component:5}
    \end{subfigure}
    \caption{Illustrations for the proof of \cref{lem:dense-component} (continued).}
    \label{fig:2tree:dense-component:secondpart}
\end{figure}
For the sake of simplicity, we omit the superscript $xy$ and write $s_{i,j}$ and $E_1$ instead of $s_{i,j}^{xy}$ and $E_1^{xy}$. Note that $G^{xy}$ consists of four copies of the sunflower graph $S$, each contributing $16$ edges, together with the edge $xy$. Hence, $|G^{xy}| = 4 \cdot 16 + 1 = 65$.
We first identify a forbidden subgraph. Observe that $G^{xy}$ contains a  subgraph isomorphic to $K_{2,4}$ where $x$ and $y$ are the part of cardinality $2$ and $s_{1,5}$, $s_{2,5}$, $s_{3,5}$ and $s_{4,5}$ are the part of cardinality $4$ (see \cref{fig:2tree:dense-component:1}). Hence, w.l.o.g., one of the edges $xs_{1,5}$ and $s_{1,5}y$ cannot belong to $E_1$. In fact, it is not sufficient if only one of these edges is not belonging to $E_1$ -- for edge $xs_{1,5}$ there are four edge-disjoint subdivisions $(x,s_{1,3},s_{1,5})$, $(x,s_{1,2},s_{1,3},s_{1,4},s_{1,5})$, $(x,s_{1,7},s_{1,5})$, $(x,s_{1,8},s_{1,7},s_{1,6},s_{1,5})$ and for $s_{1,5}y$ there is one subdivision $(s_{1,5},s_{1,9},y)$ for which we also are not allowed to add to $E_1$ one edge to destroy the induced $K_{2,4}$. Hence, if we decide that $xs_{1,5} \notin E_1$, we actually have to decide for $5$ edges -- entirely within $S_1$ -- to not be in $E_1$. Similarly, if we decide that $s_{1,5}y \notin E_1$, we actually have to decide for $2$ edges within $S_1$ ($s_{1,5}y$ and either $s_{1,5}s_{1,9}$ or $s_{1,9}y$) to not be in $E_1$.

    After resolving the $K_{2,4}$, we are still left with a forbidden subgraph isomorphic to $K_{2,3}$ where $x$ and $y$ are the part of cardinality $2$ and $s_{2,5}$, $s_{3,5}$ and $s_{4,5}$ are the part of cardinality $4$ (see \cref{fig:2tree:dense-component:2}). By the same argument as above, we may conclude that either 
    (i)~$2$ or $5$ edges  from $S_2$ should not belong to $E_1$ or (ii)~$2$ or $5$ edges from w.l.o.g. $S_3$ ($S_4$ is isomorphic) should not belong to $E_1$. At this point we observe that it is not possible that $5$ edges are missing both from $S_1$ and $S_2$ or $S_3$ since otherwise $|E_1|\leq 55$. Hence, by symmetry, we may assume w.l.o.g. that $2$ edges of either $S_2$ or $S_3$ must be missing in $E_1$.
    Next consider two subcases: (i)~not all edges of $S_2$ belong to $E_1$ and (ii)~not all edges of $S_3$ belong to $E_2$. 

    In Case~(i), we observe that there are still two subgraphs isomorphic to a $K_{2,3}$ subdivision, namely  (see also \cref{fig:2tree:dense-component:3}):
    \begin{itemize}
        \item $K_{2,3}^3$: Vertices $s_{3,5}$ and $y$ form the part of cardinality $2$ and $s_{3,3}$, $s_{3,7}$ and the path $(x,s_{4,5})$ form the part of cardinality $3$, and
        \item $K_{2,3}^4$: Vertices $s_{4,5}$ and $y$ form the part of cardinality $2$ and $s_{4,3}$, $s_{4,7}$ and the path $(x,s_{3,5})$ form the part of cardinality $3$.
    \end{itemize} We still have to resolve both $K_{2,3}^3$ and $K_{2,3}^4$. We observe that again, any edge that we can decide not to include within $E_1$ comes with a subdivision from which we have to exclude another edge. Thus, we must decide for two edges of $K_{2,3}^3$ (and its subdivision) and for two edges of $K_{2,3}^4$ (and its subdivision) to not be included in $E_1$ (potentially, this sets of edges coincide).     
    Since we have already identified four or seven edges to-not-be-included in $E_1$, we can only decide for  at most five more edges not to be included in $E_1$ and hence we are left with two options: (a)~we remove the two edges $xs_{3,5}$ and $xs_{4,5}$ and for each of them also one of their subdivision edges from $E_1$, or (b)~we decide for an edge (and one of its subdivision edges) in the intersection of $K_{2,3}^3$ and $K_{2,3}^4$. Observe that in Subcase~(a), the subgraph induced by the edges of $E_1$ is already outerplanar and we obtain one potential classification of the lemma -- more precisely, in this case, we can conclude that exactly two edges (and not five) from $S_1$ must be missing in $E_1$, as otherwise we have decided for in total $10$ edges to not belong to $E_1$. 
    In Subcase~(b), w.l.o.g. these edges are $xs_{3,5}$ and either one of $xs_{3,9}$ and $s_{3,5}s_{3,9}$. Observe that in this case, we can conclude that exactly two edges (and not five) from $S_1$ must be missing in $E_1$  as well since otherwise we would have already identified $9$ edges missing from $E_1$ without guaranteeing that the subgraph induced by $E_1$ is outerplanar.

    In Case~(ii), we observe that again there are still two subgraphs isomorphic to a $K_{2,3}$ subdivision, namely (see also \cref{fig:2tree:dense-component:4}):
    \begin{itemize}
        \item $K_{2,3}^2$: Vertices $s_{2,5}$ and $x$ form the part of cardinality $2$ and $s_{2,3}$, $s_{2,7}$ and the path $(y,s_{4,5})$ form the part of cardinality $3$, and
        \item $K_{2,3}^4$: Vertices $s_{4,5}$ and $y$ form the part of cardinality $2$ and $s_{4,3}$, $s_{4,7}$ and the path $(x,s_{2,5})$ form the part of cardinality $3$.
    \end{itemize} We conclude that we have to resolve both $K_{2,3}^2$ and $K_{2,3}^4$. We observe that once more, any edge that we can decide not to include within $E_1$ comes with a subdivision from which we have to exclude another edge. Thus, we must decide for two edges of $K_{2,3}^2$ (and its subdivision) and for two edges of $K_{2,3}^4$ (and its subdivision) to not be included in $E_1$. As in Case~(i), we have already identified four or seven edges to-not-be-included in $E_1$, and thus we can add at most five more (otherwise $|E_1| \leq 56$). Therefore -- once more -- we must decide for either: (a)~$xs_{4,5}$ and $ys_{2,5}$ and one of their subdivision or (b)~an edge (and one of its subdivision edges) in the intersection of $K_{2,3}^2$ and $K_{2,3}^4$. In Subcase~(a), we again obtain the same outerplanar subgraph as in Subcase~(a) of Case~(i) and conclude once more that exactly two edges (and not five) from $S_1$ must be missing in $E_1$. Moreover,  in Subcase (b), these edge are one of the following two choices: (b\textsubscript{1})~$xs_{2,5}$ and either one of $xs_{2,9}$ and $s_{2,5}s_{2,9}$ or (b\textsubscript{2})~$xs_{4,5}$ and either one of $xs_{4,9}$ and $s_{4,5}s_{4,9}$. Again, in both variants, exactly two edges (and not five) from $S_1$ must be missing in $E_1$ as otherwise we have already decided for nine edges to be missing from $E_1$ without yielding that the subgraph induced by $E_1$ is outerplanar.

    Summarizing the analysis so far, we have established the Constraints~\ref{prop:dense:1} to~\ref{prop:dense:4} and~\ref{prop:dense:6a} and~\ref{prop:dense:6b} and shown that if in Constraint~\ref{prop:dense:5} we decide for the case where both $s_{2,5}y \notin E_1$ and $xs_{4,5} \notin E_1$, then we obtain an outerplanar subgraph with at most $57$ edges.
    That is, to establish the lemma, it remains to argue that also $xy\notin E_1$ if $s_{2,5}y \in E_1$ or $xs_{4,5} \in E_1$. 
    
    To this end, assume w.l.o.g. that $S_4$ is the only sunflower subgraph for which we so far have not decided yet to exclude any edge from $E_1$. We observe that we still have an unresolved $K_{2,3}$, namely, the one where the part of cardinality $2$ consists of the vertices $s_{4,5}$ and $y$ and the part of cardinality $3$ consists of the the vertices $s_{4,3}$, $s_{4,7}$ and $x$ (see \cref{fig:2tree:dense-component:5}). The only edge of this $K_{2,3}$ that is not subdivided is $xy$, and the only edge of this $K_{2,3}$ that has only one edge-disjoint subdivision is $xs_{4,5}$. Hence, if we want to exclude only at most $8$ edges in total from $E_1$, we must have either (a)~$xy\notin E_1$ or (b)~$xs_{4,5}\notin E_1$. In Subcase~(a), we have successfully shown that $xy\notin E_1$, whereas in Subcase~(b), we again arrive at the same outerplanar subgraph as in the previous part of the proof (based on deciding that both both $s_{2,5}y \notin E_1$ and $xs_{4,5} \notin E_1$ in Property~\ref{prop:dense:5}). This concludes the proof.
\end{proof}

        \begin{figure}[t]
\includegraphics[width=\textwidth,page=14]{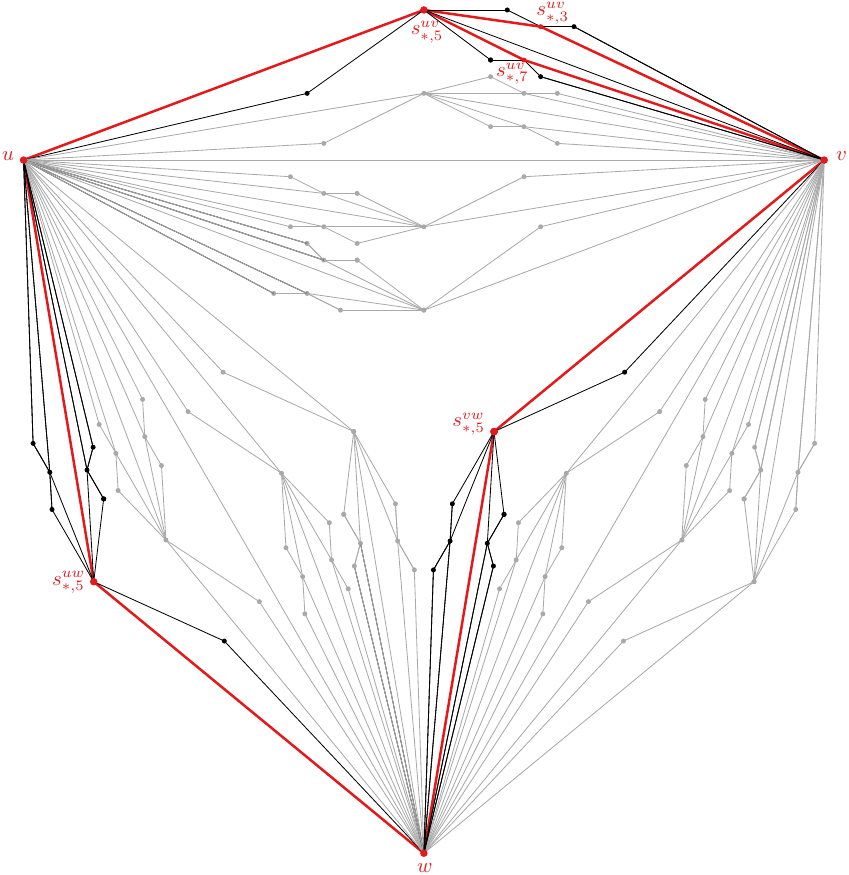}

    \caption{Illustration for the proof of \cref{lem:sparse-component}.}
    \label{fig:2tree:sparse-component}
\end{figure}

We next leverage \cref{lem:dense-component} to a stronger result for the entire graph $G$.

\begin{lemma}\label{lem:sparse-component}
    We have~$|E_1^*|\leq 172$. Moreover, if~$|E_1^*| = 172$, then, up to renaming of isomorphic components:
    \begin{enumerate}[label=\textsf{\textbf{\textcolor{lipicsGray}{\upshape{{\arabic*}.}}}},ref=\textsf{\textbf{\textcolor{lipicsGray}{\upshape{{\arabic*.}}}}}]
        \item\label{prop:sparse:1} $G^{uv}$ and $G^{vw}$ each contain $57$ edges of $E_1^*$ according to the constraints of \cref{lem:dense-component}, i.e., in particular, $uv\notin E_1^*$ and $vw\notin E_1^*$
        \item\label{prop:sparse:2} $uw, s_{1,5}^{uw}w, s_{2,5}^{uw}w, s_{3,5}^{uw}u, s_{4,5}^{uw}u \notin E_1^{*}$ (5 missing edges)
        \item\label{prop:sparse:3} for $i\in \{1,2\}$ either $s_{i,5}^{uw}s_{i,9}^{uw} \notin E_1^{*}$ or $s_{i,9}^{uw}w \notin E_1^{*}$ (total of 2 missing edges) 
        \item\label{prop:sparse:4}  for $i\in \{3,4\}$ either $s_{i,5}^{uw}s_{i,9}^{uw} \notin E_1^{*}$ or $s_{i,9}^{uw}u \notin E_1^{*}$ (total of 2 missing edges).
    \end{enumerate}
    Finally, if $|E_1^*| \in \{170,171\}$, then:
    \begin{enumerate}[label=\textsf{\textbf{\textcolor{lipicsGray}{\upshape{\Alph*.}}}},ref=\textsf{\textbf{\textcolor{lipicsGray}{\upshape{\Alph*.}}}}]
 \item \label{prop:sparse:A} $G^{uv}$ contains $58$ edges of $E_1^*$ according to the constraints of \cref{lem:dense-component},  $G^{vw}$ contains at least $56$  edges of $E_1^*$ according to the constraints of \cref{lem:dense-component} and $G^{wu}$ contains at most $56$ edges according to the Constraints \ref{prop:sparse:2} to~\ref{prop:sparse:4} of this lemma, or
 \item \label{prop:sparse:B} $G^{uv}$ and $G^{vw}$  each contain $57$ edges of $E_1^*$ according to the constraints of \cref{lem:dense-component} and $G^{wu}$ contains $56$ or $57$ edges of $E_1^*$ according to the constraints of \cref{lem:dense-component}.
    \end{enumerate}
\end{lemma}

\begin{proof}
    First, we consider the case where the induced subgraph $G^{xy}[E_1^*]$ for one pair $x,y\in \{u,v,w\}$ contains no path between $x$ and $y$ and derive an upper bound for $|E_1^*|$ in this scenario. Observe that each $G^{xy}$ with 
    $x,y\in\{u,v,w\}$ and $x\neq y$ contains nine edge-disjoint paths between $x$ and $y$:
    \begin{enumerate*}
        \item $(x,y)$,
        \item[\textsf{\textbf{\textcolor{lipicsGray}{\upshape{2.-5.}}}}] $\forall i \in\{1,2,3,4\}: (x,s_{i,5}^{xy},y)$,
        \item[\textsf{\textbf{\textcolor{lipicsGray}{\upshape{6.-7.}}}}] $\forall i \in\{1,2\}: (x,s_{i,3}^{xy},\allowbreak s_{i,5}^{xy},s_{i,9}^{xy},y)$, and
        \item[\textsf{\textbf{\textcolor{lipicsGray}{\upshape{8.-9.}}}}] $\forall i \in\{3,4\}: (y,s_{i,3}^{xy},s_{i,5}^{xy},s_{i,9}^{xy},x)$.
    \end{enumerate*}
    Hence, if the induced subgraph $G^{xy}[E_1^*]$ for one pair $x,y\in \{u,v,w\}$ contains no path between $x$ and $y$, then $|E_1^*|\leq 172$ follows immediately by applying \cref{lem:dense-component} to the two subgraphs with no induced $xy$-path, say  $G^{uv}$ and $G^{vw}$, and removing one edge from each edge-disjoint $xy$-path in the remaining component, say $G^{uw}$.
    Moreover, the only way to decide for $9$ edges to not belong to $E_1^*$, while still ensuring that $G^{xy}[E_1^*]$ contains no $xy$-path,
    is by adhering to Constraints~\textsf{\textbf{\textcolor{lipicsGray}{\upshape{2.-4.}}}} of the statement of the lemma as the edges incident to $s_{i,3}$ can be replaced with symmetric edges incident to $s_{i,7}$;
    i.e., not assigning them to $E_1^*$ does not decrease the number of induced $xy$-paths in $G^{xy}[E_1^*]$.

It remains to consider the scenario where the induced subgraph $G^{xy}[E_1^*]$ for each pair $x,y\in \{u,v,w\}$ contains a path between $x$ and $y$. We first argue that in this scenario, indeed we even have that $|E_1^*|<172$. Hence, assume for a contradiction that $|E_1^*|\geq 172$ 
and that all of $G^{xy}$ with $x,y\in\{u,v,w\}$ and $x \neq y$ contain an induced $xy$-path. Since $|E_1^*|\geq 172$, by \cref{lem:dense-component}, there must be at least one $G^{xy}$, say w.l.o.g., $G^{uv}$, such that $G^{uv}[E_1^*]$ contains the edges (see \cref{fig:2tree:sparse-component}): 
    $us_{*,5}^{uv}$,
     $s_{*,3}^{uv}s_{*,5}^{uv}$,
    $s_{*,7}^{uv}s_{*,5}^{uv}$,
     $vs_{*,3}^{uv}$,
     $vs_{*,7}^{uv}$ with $*\in\{1,2,3,4\}$.
But then, there is an induced subgraph with a $K_{2,3}$ in $G[E_1^*]$; namely:
\begin{itemize}
    \item $s_{*,5}^{uv}$ and $v$ form one part of the vertex set
    \item $s_{*,3}^{uv}$, $s_{*,7}^{uv}$ and the contraction of $p_{vw}p_{wu}$ form the other part of the vertex set where $p_{vw}$ and $p_{wu}$ are the $vw$-path in $G^{vw}[E_1^*]$ and the $wu$-path in $G^{wu}[E_1^*]$; respectively.
\end{itemize}
This contradicts the outerplanarity of $G[E_1^*]$.

\noindent Therefore, in the following, we have $|E_1^*|\in\{170,171\}$. We consider two subcases: (\ref{prop:sparse:A})~one subgraph, w.l.o.g., $G^{uv}$ contains $58$ edges of $E_1^*$ and (\ref{prop:sparse:B})~no subgraph $G^{xy}$ contains $58$ edges of $E_1^*$. 
In Subcase (\ref{prop:sparse:A}), we observe that if one of the subgraphs, w.l.o.g., $G^{vw}[E_1^*]$, contains an induced $vw$-path, then by the previous line of argument for the case $|E_1^*|=172$, subgraph $G^{uw}$ must adhere to Constraints~\ref{prop:sparse:2}-\ref{prop:sparse:4}. Note that in this particular  case, if $G^{vw}$ contains at least $57$ edges of $E_1^*$. This provides the Constraint~\ref{prop:sparse:A}.

In Subcase (\ref{prop:sparse:B}), we observe that each subgraph contains at least $56$ edges of $E_1^*$ and thus Constraint~\ref{prop:sparse:B} follows immediately from \cref{lem:dense-component}.
\end{proof}

\subparagraph*{Weighted Book Thickness of $G$.} We obtain a $3$-page book embedding $L^*$ of $G$ with weighted book thickness $\frac{220}{195}$ using the planar embedding and compatible thick yellow subhamiltonian cycle $C$ shown in \cref{fig:optimal-2tree-embedding} where the dark gray edges within the yellow shaded region are assigned to Page~$1$, the blue edges are assigned to Page~$2$ and the two red edges are assigned to Page~$3$. Observe that transforming $C$ to the spine results in the valid $3$-page book embedding $L^*$ as the two edges on Page~$3$ do not intersect and that we have exactly the edges propagated by \cref{lem:sparse-component} not on Page~$1$. In the following, we use the computer-aided part of the proof to 
\begin{enumerate}[label=\textsf{\textbf{\textcolor{lipicsGray}{\upshape{{\arabic*.}}}}},ref=\textsf{\textbf{\textcolor{lipicsGray}{\upshape{{\arabic*}}}}}]
    \item\label{computerproof:1} show that for any $2$-page book embedding of $G$, at most $169$ edges are assigned to Page~$1$
    \item\label{computerproof:2} establish the optimality of $L^*$ with respect to the weighted book thickness.
\end{enumerate}

\begin{figure}[t]
    \centering
    \includegraphics[width=\linewidth,page=2]{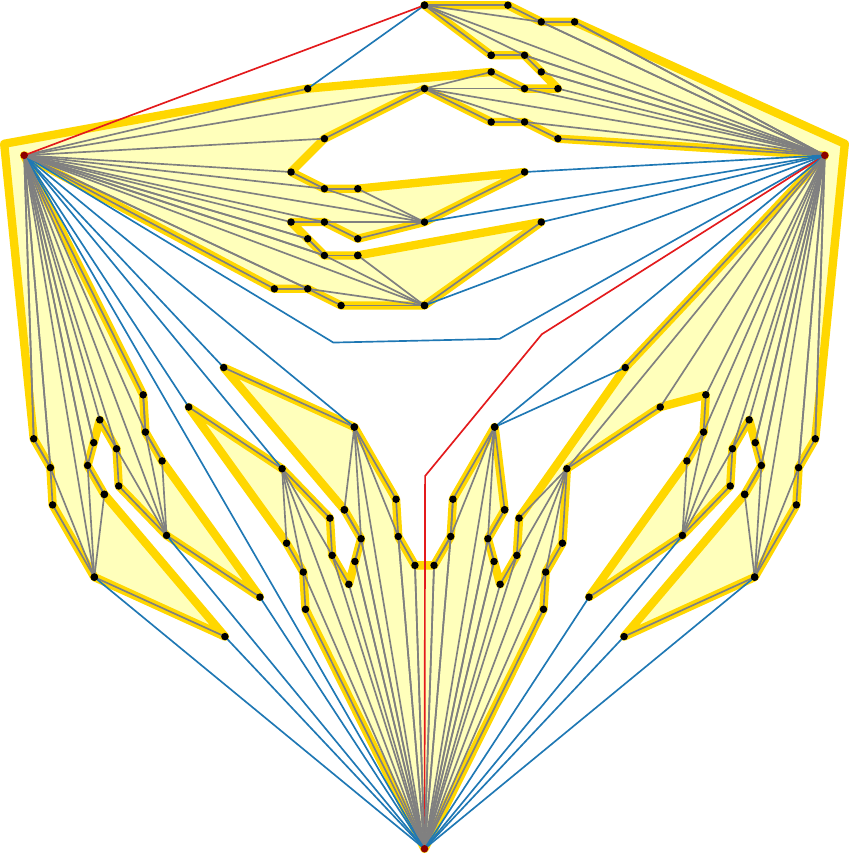}
    \caption{Embedding of $G$ with optimal weighted book thickness.}
    \label{fig:optimal-2tree-embedding}
\end{figure}

\subparagraph*{The computer-aided part of the proof.} In the previous theoretical analysis, we focused on structural properties of graph $G$. In order to establish \cref{thm:2tree}, it is additionally necessary to evaluate the topology of linear layouts of $G$ which have at least $|E_1^*| \geq 170$ edges on Page~$1$. To avoid a theoretical analysis of an excessive amount of subcases, we complement our  structural analysis of forbidden subdivisions with a SAT solving approach. To this end, we  encode the constraints described in~\cref{lem:dense-component,lem:sparse-component} within the SAT framework by Bekos, Kaufmann, and Zielke~\cite{DBLP:conf/gd/Bekos0Z15} and check for embeddability with $2$ and $3$ pages, respectively. 
Their SAT formulation uses the following boolean variables~\cite{DBLP:journals/corr/abs-2003-09642,DBLP:conf/gd/Bekos0Z15}:
\begin{itemize}
    \item $\sigma(u,v) = \texttt{true}$ if and only if for a pair of vertices $u$ and $v$, $u$ precedes $v$ on the spine,
    \item $\phi_p(e)$ = $\texttt{true}$ if and only if edge $e$ is assigned to page $p$, and,
    \item $\chi(e,e')$ = $\texttt{true}$ if and only if edges $e$ and $e'$ are assigned to the same page. 
\end{itemize}
At a very high level, constraints in the SAT-formulation make use of (i)~the variables~$\sigma$ to define a linear order of the vertices, (ii)~variables $\phi$ to ensure that each edge is assigned to exactly one page and (iii)~variables $\chi$, whose value is computed based on the values of variables $\phi$,  to ensure which edges must not intersect using the variables~$\sigma$. We ran the following experiments on a single-node 8-core 2.5 GHz Intel Core i5-12400 machine with 16GB RAM. The source codes of both of our  experiments is publicly available at \supplementlink.

\subparagraph*{Experiment \ref{computerproof:1}: 2-Page Embeddability with at least 170 Edges on Page 1.} In this experiment, we performed three distinct subexperiments. First, we encoded only Constraints \ref{prop:sparse:1}-\ref{prop:sparse:4} of \cref{lem:sparse-component} and tested for $2$-page embeddability. Recall that these constraints capture the local restriction for each component $G^{xy}$ when Page~$1$ induces a maximal outerplanar subgraph of $G$. For all possible combinations of these configurations, the resulting SAT instances were unsatisfiable. 
Afterwards, we extended the previous experiment by encoding Constraints \ref{prop:sparse:A}  of \cref{lem:sparse-component} in the second subexperiment and \ref{prop:sparse:B} of \cref{lem:sparse-component} in the third subexperiment. These constraints capture the possibilities that arise when $170$ or $171$ edges are assigned to Page~$1$. Again, we enumerate all possible combinations of these constraints and encode them as SAT instances and solve them. All resulting instances were again unsatisfiable, which verifies the following:
\begin{proposition}\label{prop:2page-embeddability}
    There is no $2$-page book embedding of $G$ with at least $170$ edges on Page~$1$.
\end{proposition}

\subparagraph*{Experiment \ref{computerproof:2}: 3-page Embeddability with 1 Edge on Page 1.} Finally, we verify the optimality of the $3$-page book embedding $L^*$ of graph $G$ implied by the planar embedding and subhamiltonian cycle in \cref{fig:optimal-2tree-embedding}, which assigns the two edges $vw$ and $us_{4,5}$ to Page~$3$. In $L^*$, 172 edges have been assign to Page~$1$, 21 edges to Page~$2$ and 2 edges to Page~$3$. Thus, the corresponding weighted book thickness is $\frac{220}{195}$. According to \cref{lem:sparse-component}, the only way to improve upon $L^*$ would be to assign exactly 1 edge to Page~$3$, while again 172 edges must be assigned to Page~$1$ adhering to the constraints of \cref{lem:sparse-component}. We encoded the corresponding constraints in a SAT formula and checked for satisfiability of all possible variations of the constraints. Since again none of the generated SAT instances were satisfiable, we conclude:

\begin{proposition}\label{prop:3page-embeddability}
    There is no $3$-page book embedding of $G$ with $172$ edges on Page~$1$, $22$ edges on Page~$2$ and $1$ edge on Page~$3$. 
\end{proposition}

\subparagraph{Conclusion.} \cref{prop:2page-embeddability} implies that $\mathrm{wbt}_2(G) \geq \frac{221}{195}$ as every $2$-page book embedding has at most $169$ edges on Page~$1$ and at least $26$ edges on Page~$2$. In contrast, \cref{prop:3page-embeddability} and our construction in \cref{fig:optimal-2tree-embedding} imply that $\mathrm{wbt}_3(G) = \frac{220}{195}$. This proves \cref{thm:2tree}.

\section{Graphs of Pathwidth 2}
\label{sec:pathwidth2}
Given that already for planar graphs of pathwidth~3 (\cref{thm:3-path-distinct}) and for graphs of treewidth~2 (\cref{thm:2tree}),
two pages may not be sufficient to allow for a minimum-weight book embedding,
we restrict our class of graphs even further
to graphs of pathwidth 2 where two pages are always sufficient.

\begin{theorem}
    \label{thm:pathwidth2}
    Let $G$ be a graph of pathwidth at most~2.
    Then, $\wbt(G) = \wbt_2(G)$.
    Moreover, a book embedding of weight $\wbt(G)$
    can be computed in linear time.
\end{theorem}

\begin{proof}
Recall that the set of edges on any page corresponds to an outerplanar graph.
Hence, to show the theorem,
it suffices to show that there is a maximum-size outerplanar subgraph of~$G$,
which we will put on the first page with an order~$\pi$ on the spine,
such that the remaining edges can be put on the second page without crossings
given the order~$\pi$ on the spine.

W.l.o.g., assume that $G$ is biconnected;
otherwise, decompose $G$ into its biconnected components in linear time~\cite{DBLP:journals/cacm/HopcroftT73},
obtain book embeddings for each biconnected component independently
and put them together afterwards.
To this end, consider a cut vertex
being part of at least two biconnected components.
Reserve for the remaining vertices of each such biconnected component
a contiguous interval along the spine.
Observe that the edges of two distinct biconnected components
can nest but cannot cross.

Bar{\'{a}}t, Hajnal, Lin and Yang~\cite[Theorem~3.1]{structurePathWidth2}
characterize biconnected graphs with pathwidth at most~2
as graphs consisting of two vertex-disjoint paths $P$ and $Q$
that are connected by an arbitrary number of edges and paths of length~2
(we call them \emph{chords} and \emph{subdivided chords}, respectively)
such that:
\begin{itemize}
    \item if we draw~$P$ and~$Q$ on two horizontal lines,
    we can draw all (subdivided) chords with straight-line segments between those lines without crossings, and
    \item the leftmost vertices of~$P$ and~$Q$ are connected by a (subdivided) chord and
    the rightmost vertices of~$P$ and~$Q$ are connected by a (subdivided) chord.
\end{itemize}
We say that a path~$S$ of length~2 is parallel to an edge~$e$
(or another path~$S'$ of length~2)
if $S$ connects the two endpoints of~$e$ (or $S'$).
Note that the input graph~$G$ does not come with a labeling
of the edges as belonging to $P$ or~$Q$ or being a chord or subdivided chord.
However, we can easily obtain this as we show by the following claim.
For an illustration, see \cref{fig:graph-pw2}.

\begin{figure}[p]
    \centering
    \begin{subfigure}{\linewidth}
        \centering
        \includegraphics[page=1]{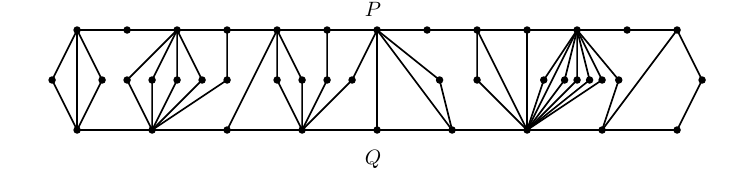}
        \subcaption{Structure of the input graph~$G$.}
        \label{fig:structure-pw2}
    \end{subfigure}
    
    \medskip
    
    \begin{subfigure}{\linewidth}
        \centering
        \includegraphics[page=2]{2-paths}
        \subcaption{Groups of prallel chords and subdivided chords.}
        \label{fig:groups-pw2}
    \end{subfigure}
    
    \medskip
    
    \begin{subfigure}{\linewidth}
        \centering
        \includegraphics[page=3]{2-paths}
        \subcaption{Faces after removing parallel chords and subdivided chords of the groups.}
        \label{fig:faces-pw2}
    \end{subfigure}
    
    \medskip
    
    \begin{subfigure}{\linewidth}
        \centering
        \includegraphics[page=4]{2-paths}
        \subcaption{Marked edges in~$M$ (blue and thick),
        removed to obtain a maximum-size outerplanar subgraph.}
        \label{fig:marking-pw2}
    \end{subfigure}
    
    \medskip
    
    \begin{subfigure}{\linewidth}
        \centering
        \includegraphics[page=5]{2-paths}
        \subcaption{Subhamiltonian path (light red) separting the edges of the two pages (black vs.\ blue edges).}
        \label{fig:subhamiltonian-pw2}
    \end{subfigure}
    
    \caption{Example of a biconnected graph~$G$ with pathwidth~2 and how we determine a 2-page book embedding
        realizing~$\wbt(G)$.}
    \label{fig:graph-pw2}
\end{figure}

\begin{claim}
    \label{clm:2path-structure}
    Let~$G$ be a biconnected graph of pathwidth at most~2.
    There is a linear-time algorithm that determines
    $P$, $Q$, the chords and subdivided chords.
    Moreover, the algorithm yields a grouping of parallel chords and subdivided chords,
    and a sequence of faces $f_1, \dots, f_t$ (for some $t \ge 0$)
    arising in a drawing of~$G$ where~$P$ and~$Q$ are arranged along two horizontal lines
    and, for each each group of parallel (subdivided) chord, there is one
    representative drawn with straight-line segments between~$P$ and~$Q$;
    see \cref{fig:groups-pw2,fig:faces-pw2}.
\end{claim}

\begin{claimproof}
    If $G$ is just an edge, a path of length~2, or a cycle, the theorem holds trivially.
    Excluding these cases, we can assume that the first and last vertices of $P$ and~$Q$
    are connected by at least one subdivided chord each;
    otherwise, if they were just connected by a chord,
    the first vertex of $P$ or~$Q$ would have degree~2.
    We could re-define $P$ and~$Q$ such that this degree-2 vertex is the middle vertex of a subdivided chord.
    
    Now observe that, if there is a path~$S$ of length~2,
    whose middle vertex has degree~2 in~$G$,
    such that $S$ is parallel to an edge~$e$,
    then $e$ is a chord and $S$ is a subdivided chord.
    To mark this configuration, iterate in linear time over each degree-2 vertex~$v$
    and check if the length-2 path whose middle vertex is~$v$ is parallel to an edge.
    If so, label $S$ as \textsf{subdivided chord}, label~$e$ as \textsf{chord},
    and remove $e$ for now.
    
    Similarly, if there is a path~$S$ of length~2,
    whose middle vertex has degree~2 in~$G$,
    such that $S$ is parallel to another such path~$S'$ of length~2,
    then both $S$ and~$S'$ are subdivided chords (since we have excluded the case of~$G$ being just a cycle).
    To mark this configuration, label each vertex $v$ with a unique number~$\ell(v)$ from 1 to $n$.
    Then, create an empty list $L$.
    Iterate in linear time over each degree-2 vertex~$v$ of~$G$ where $v$ has the two neighbors~$u$ and~$w$.
    Insert a new element into~$L$, which has the pair of numbers $(\ell(u), \ell(w))$ as its key,
    and has $v$ as its value.
    After filling $L$, sort $L$ in linear time by its keys using, e.g., RadixSort.
    Now iterate in linear time through $L$ and check each pair of consecutive elements.
    If the two keys are identical, we have found a pair of parallel subdivided chords.
    We use the values~$v$ and~$v'$ of these elements
    to identify the corresponding paths of length~2, which we both label \textsf{subdivided chord}
    and contract one into the other for now.
    
    After all these (temporary) removals of parallel chords and subdivided chords,
    consider what a drawing of $P$ and $Q$ along two horizontal lines would look like:
    it forms a sequence~$f_1, \dots, f_t$ of faces (for some integer $t$)
    that are separated by (non-parallel) chords or subdivided chords;
    see \cref{fig:faces-pw2}.
    
    We determine the faces and in this process label all edges.
    We use the colors black for edges whose incident faces have been determined
    and white for all other edges.
    Initially all edges are white.
    We can assume $t \ge 2$ as otherwise, $G$ would be a cycle or a path.
    Call a vertex of degree at least~3 in the reduced graph a \emph{high-degree vertex}.
    Note that the faces $f_2, \dots, f_{t-1}$ have three or four high-degree vertices,
    while the faces $f_1$ and~$f_t$ have exactly two high-degree vertices.
    Using this property, we determine $f_1$ and~$f_t$ in linear time as follows.
    Consider each maximal path~$X$ whose inner vertices have degree 2 in~$G$;
    if the endpoints of the path are connected by an edge~$e$,
    then $X$ and $e$ form $f_1$ or $f_t$
    (see, e.g., $f_{14}$ in \cref{fig:faces-pw2}).
    We label $e$ \textsf{chord} and assign the edges of~$X$ to $P$ and~$Q$.
    If two edges in~$X$ are labeled already \textsf{subdivided chord},
    the predecessors and successors belong to~$P$ and~$Q$, respectively;
    if not, label two consecutive edges \textsf{subdivided chord}
    and assign the other edges to $P$ and~$Q$ accordingly.
    We proceed similarly if the endpoints of $X$ are connected
    by a path~$S$ of length~2 such that middle vertex of~$S$ has degree~2
    (see, e.g., $f_1$ in \cref{fig:faces-pw2}).
    We label $S$ as \textsf{subdivided chord} and otherwise proceed in the same way.
    This can be done in linear time since, for these checks,
    we traverse each edge at most a constant number of times.
    Color all edges of~$f_1$ and~$f_t$ black and let $i = 2$.
    
    We have determined the chord or subdivided chord separating $f_{i-1}$ and~$f_i$
    and the two corresponding high-degree vertices~$u$ and~$w$.
    While there are white edges, we want to determine the edges bounding~$f_i$.
    At least one of $u$ and~$w$, say $u$, has exactly one incident white edge~$e$.
    Then, label $e$ with $P$ or~$Q$ (depending on the labeling around~$u$ in the previous faces).
    While the other endpoint of~$e$ has degree~2,
    let the next edge be~$e$ and assign it also to~$P$ or~$Q$.
    Eventually, we arrive at a high-degree vertex~$u'$.
    Proceed similarly at~$w$ if $w$ has exactly one incident white edge
    and let $w'$ be the vertex where this process stops (potentially $w = w'$).
    If $u'$ and $w'$ are connected by an edge~$e$, label $e$ as \textsf{chord}.
    Otherwise, $u'$ and $w'$ are connected by a path~$S$ of length~2
    whose middle vertex has degree~2.
    To check this, search starting from $u'$ or $w'$ depending on which of these two vertices has fewer incident white edges.
    Label $S$ as \textsf{subdivided chord}.
    Color the edges of $f_i$ black and set $i = i + 1$.
    
    This process of determining the edges of the next face
    and adding labels can be done in linear time since we only check (white) degrees
    of vertices and do not traverse the same edge multiple times.
    We can re-add the removed edges and save the resulting groups of parallel (subdivided) chords.
\end{claimproof}

Next, we mark edges that we will put onto page~2
by adding them to a set~$M$.
For an illustration, see \cref{fig:marking-pw2}.
We will show afterwards that the unmarked edges correspond to a maximum-size outerplanar subgraph of~$G$.

Traverse the groups $\s_1, \dots, \s_{t+1}$
of parallel (subdivided) chords in the order they are separating $f_1, \dots, f_t$.
Note that $\s_1$ and $\s_{t+1}$ separate $f_1$ and $f_t$, respectively, from the outer face, to which we refer with $f_0$ or~$f_{t+1}$.
Let $i = 1, \dots, t+1$ and consider~$\s_i$.
The faces $f_{i-1}$ and $f_i$ are separated by the (subdivided) chords in~$\s_i$.
Each such face contains at least one edge from $P$ or~$Q$.
Let $R_{i-1} = P$ if $f_{i-1}$ contains an edge from~$P$ and let $R_{i-1} = Q$ otherwise.
Define $R_i$ analogously with respect to~$f_i$.
If $\s_i$ contains more than two subdivided chords,
add from all but two of these subdivided chords the edge
having an endpoint on $R_i$ to~$M$.
If $\s_i$ contains at least two subdivided chords,
add the first edge on $R_i$ bounding $f_i$ to~$M$ unless $i = t+1$, and
add the first edge on $R_{i-1}$ bounding $f_{i-1}$ to~$M$ if not done already and unless $i = 1$.
If $\s_i$ contains exactly one subdivided chord
and $i \ne 1$ and $i \ne t + 1$,
then check if an edge on $R_{i-1}$ bounding $f_{i-1}$
has already been added to~$M$;
if not, add the first edge from $R_i$ bounding $f_i$ to~$M$.
Clearly, this procedure to determine~$M$ runs in linear time.

To show that the edges in $E(G) \setminus M$
can be assigned to the first page, while
the edges in~$M$ can be assigned to the second page,
we next describe a planar embedding of~$G$ with a planar subhamiltonian cycle~
$C$ such that all edges of~$E(G) \setminus M$ lie on~$C$ or on the inside of~$C$,
while all edges of~$M$ lie on the outside of~$C$.
The ordering of $V(G)$ on the spine follows~$C$.

The drawing of~$G$ follows from the previous description of~$P$ and~$Q$,
only the order~of parallel (subdivided) chords must be chosen;
it follows implicitly from the following description to keep~$C$ planar.
We describe an upper and lower contour of~$C$ simultaneously
by traversing the faces from left to right
(i.e., by index $i = 1, \dots, t+1$).
Both contours start in the middle vertex of a subdivided chord~$S \in \s_1$
where no edge of~$S$ is in~$M$.
See \cref{fig:subhamiltonian-pw2} for an illustration of~$C$.
The upper contour follows the vertices of~$P$,
while the lower contour follows the vertices in~$Q$.
Suppose face $f_i$ has an edge~$e = uv$ of $P$ or~$Q$ in~$M$.
Recall that $e \in R_i$, where $R_i \in \{P, Q\}$.
In $C$, we add between $u$ and $v$
all middle vertices of subdivided chords of~$\s_i$
that have not been added to~$C$ yet
(at most one has been added when we were considering~$f_{i-1}$
as described next).
The order of corresponding subdivided chords is such that
the ones without edges in~$M$ precede those with an edge in~$M$.
Then, before the contour continues with $v$
along~$P$ or~$Q$, we add a middle vertex
of a subdivided chord~$S \in \s_{i+1}$
where no edge of~$S$ is in~$M$ (if existent;
otherwise the contour continues directly with~$v$).
We close the cycle $C$ by the middle vertices of~$\s_{t+1}$
(except for potentially one that has been added to the contour
when considering~$f_t$).
The order of the middle vertices along $C$ in this traversal
prescribes the order of the (subdivided) chords in the corresponding planar drawing.
Note that~$C$ is a planar subhamiltonian cycle
that separates the edges in~$E(G) \setminus M$ from those in~$M$;
in particular, we have added in a face
the first edge of $P$ or $Q$ to~$M$
and edges of~$M$ belonging to subdivided chords
are adjacent to them,
which means that the edges of $M$ do not
cross the subhamiltonian cycle.

It remains to show that $E(G) \setminus M$ induces
a maximum-size outerplanar subgraph of~$G$.
\begin{claim}
    \label{clm:nok23}
    The graph~$G$ becomes outerplanar when the edges of~$M$ are removed.
\end{claim}

\begin{claimproof}
As a certificate consider the drawing (implicitly) described before.
All vertices of~$P$ and~$Q$ lie on the outer face.
Consider each group $\s_i$ of parallel (subdivided) chords.

If $\s_i$ contains two or more subdivided chords,
all but two subdivided chords have become leaves;
call the two remaining subdivided chords $S$ and~$S'$.
Place~$S$, leftmost, i.e.,
incident to face~$f_{i-1}$.
Since we have removed an edge of~$P$ or~$Q$ from~$f_{i-1}$,
$f_{i-1}$ is merged into the outer face and the middle vertex of~$S$
lies on the outer face.
Next, arrange a potential chord of~$\s_i$,
which does not need to lie on the outer face.
Then, add~$S'$ and all leaves such that they are incident to~$f_i$.
Since we have also removed an edge of~$P$ or~$Q$ from~$f_i$,
$f_i$ is now also merged into the outer face, and
the middle vertex of~$S'$ and all leaves lie on the outer face.

If $\s_i$ contains one subdivided chord,
we have removed an edge of~$P$ or~$Q$ from~$f_{i-1}$ or~$f_i$.
Place the subdivided chord incident to that face
as it is now merged into the outer face.
Hence, the middle vertex lies on the outer face.
If there is just a chord, 
there is no vertex not on~$P$ or~$Q$
that needs to lie on the outer face.
\end{claimproof}

To show maximality, it suffices to argue about edge sets whose
removal eliminates all $K_{2,3}$ minors
since $K_4$ has pathwidth~3 and pathwidth is closed under taking minors.

\begin{claim}
    \label{clm:min-marking}
    For any set~$M^\star \subseteq E(G)$ such that $G$ has no $K_{2,3}$ minor when the edges of~$M^\star$ are removed,
    it holds that $|M^\star| \ge |M|$.
\end{claim}

\begin{claimproof}
    Suppose for a contradiction that there is a set of edges $M^\star$
    whose removal results in a maximum-size outerplanar subgraph and $|M^\star| < |M|$.
    
    Clearly, each group with more than two subdivided chords forms a $K_{2,3}$ on its own.
    Hence, like $M$, $M^\star$ contains for each subdivided chords,
    except for potentially two of them, an edge.
    Therefore, we assume in the rest of the proof
    that we have a reduced instance
    where, for any $i \in \{1, \dots, t+1\}$, $\s_i$ contains at most two subdivided chords.
    Also note that it suffices to consider biconnected components
    when searching for $K_{2,3}$ minors.
    
    The edges of~$G$ have a natural order from left to
    right given by the order of the groups:
    for each group, we consider first its subdivided chords as twins,
    then its chord.
    Between the edges of each two consecutive groups~$\s_i$ and~$\s_{i+1}$,
    we have the edges of~$P$ and~$Q$ on the boundary of the face~$f_i$.
    
    Let~$e$ be the first edge where the containment in~$M$ and~$M^\star$ differs.
    We consider several cases where $e$ lies.
    We will see that in each case, we can modify $M^\star$ slightly
    to be more similar to~$M$ without changing its size and without introducing a $K_{2,3}$ minor.
    We can repeat the following case distinction until $M^\star = M$, which proves the claim.
    
    \subparagraph{$e$ lies on the only subdivided chord of its group~$\s_i$.}
    We know, by the construction of~$M$,
    that $e$ cannot be in~$M$ and, thus, $e \in M^\star$.
    We remove $e$ from $M^\star$ and add
    an edge of~$P$ or~$Q$ on the boundary of~$f_i$
    (that is, the face to the right of the group~$\s_i$)
    to $M^\star$ instead.
    Now $\s_i$ cannot be in a biconnected component
    with anything to its right, only with things to its left.
    We know, however, that $\s_i$ and everything
    to its left is now the same for $M$ and~$M^\star$
    and, therefore, has no $K_{2,3}$ minor.
    
    \subparagraph{$e$ lies on one of two subdivided chords in its group~$\s_i$.}
    Again, $e \notin M$ but $e \in M^\star$.
    We know that both~$f_{i-1}$ and~$f_i$ have an edge of $P$ or~$Q$ in $M$, hence,
    the same edge of~$f_{i-1}$ is also in~$M^\star$.
    If also~$f_i$ had an edge of $P$ or~$Q$ in $M^\star$,
    then $\s_i$ is not in a biconnected component with anything else
    and we could simply remove $e$ from $M^\star$,
    which contradicts $M^\star$ being minimum.
    So, remove $e$ from $M^\star$ and add
    an edge of~$P$ or~$Q$ on the boundary of~$f_i$
    to $M^\star$ instead.
    Now $\s_i$ cannot be in a biconnected component
    with anything to its left or right,
    and, therefore, has no $K_{2,3}$ minor.
    Note also that among all edges of~$\s_i$,
    only $e$ could have been in $M^\star$ as otherwise
    $M^\star$ would not be minimum.
    
    \subparagraph*{$e$ is the chord in its group~$\s_i$.}
    Since $M$ and~$M^\star$ are equal until~$e$,
    no edge of the (potentially existent) subdivided chords of~$\s_i$ is in~$M^\star$.
    The group~$\s_i$ does not have a $K_{2,3}$ minor when considering~$\s_i$
    together with everything on the left.
    We can avoid a $K_{2,3}$ minor by replacing~$e$ in~$M^\star$
    with an edge from $P$ or~$Q$ on the boundary of~$f_i$
    (that is, the face to the right of the group~$\s_i$)
    since then there is no biconnected component including (parts of)~$\s_i$
    and anything to its right.
    
    \subparagraph*{$e$ lies on $P$ or~$Q$.}
    Let $f_i$ be the face on which $e$ lies.
    If there is just one edge of~$P$ or~$Q$ belonging to~$M^\star$
    and a different edge of~$P$ or~$Q$ belonging to~$M$,
    then we can simply replace the edge in~$M^\star$ by that in~$M$
    without changing the (non-trivial) biconnected components
    and, hence, without creating a $K_{2,3}$.
    By the algorithm creating $M$,
    there is at most one edge of~$P$ or~$Q$ per face in~$M$.
    Moreover, there is at most one edge of~$P$ or~$Q$ per face in~$M^\star$
    because if there were two, we could remove one and would not change
    the set of (non-trivial) biconnected components;
    this contradicts $M^\star$ being minimum.
    
    Now suppose that $e \in M^\star$ and none of the edges of~$P$ and~$Q$ on~$f_i$ are in~$M$.
    Then, we know that there is no~$K_{2,3}$
    when considering the partial graph until the group~$\s_{i+1}$
    on the direct right side of~$f_i$
    because also none of the (subdivided) chords is in~$M$.
    This means removing~$e$ from~$M^\star$ and
    adding an edge from~$P$ or~$Q$ on~$f_{i+1}$ is a safe operation
    to not add any $K_{2,3}$ minor in the graph where we remove~$M^\star$
    since everything to the right of~$\s_{i+1}$ is then in a different biconnected component.
    
    Finally, suppose that $e \in M$ and none of the edges of~$P$ and~$Q$ on~$f_i$ are in~$M^\star$.
    If an edge of~$f_i$ is in $M$, this means that there is a subdivided chord in~$\s_i$,
    and either a second subdivided chord in~$\s_i$ or no edge of~$P$ and~$Q$ on~$f_{i-1}$ is in~$M$.
    In both situations, this would result in a $K_{2,3}$ minor together with the (subdivided) chords in~$\s_{i+1}$.
    Hence, each (subdivided) chord in~$\s_{i+1}$ has an edge in~$M^\star$.
    If there are at least two (subdivided) chords in~$\s_{i+1}$,
    replace these two edges in~$M^\star$ by an edge of~$P$ or~$Q$ on the boundary of~$f_i$
    and by an edge of~$P$ or~$Q$ on the boundary of~$f_{i+1}$,
    which suffices to avoid $K_{2,3}$ minors due to the separation into distinct biconnected components.
    If there is exactly one (subdivided) chord~$S$ in~$\s_{i+1}$,
    replace its edge in~$M^\star$ by the edge of~$P$ or~$Q$ on the boundary of~$f_i$ that is in~$M$.
    Note that this does not create a new $K_{2,3}$ minor with
    the (subdivided) chord~$S$ and anything that follows
    (in the graph where we remove~$M^\star$) on the right side of~$S$
    because if there was a path from the vertex of~$S$ on~$P$
    to the vertex of~$S$ on~$Q$ on the right side, this would have formed a $K_{2,3}$ minor
    together with the two subdivided chords of~$\s_i$,
    or with the subdivided chord of~$\s_i$ and the (subdivided) chord of~$\s_{i-1}$,
    already before.
\end{claimproof}

We have found a maximum-size outerplanar subgraph
due to \cref{clm:nok23,clm:min-marking}.
As argued before,
the remaining edges do not cross the subhamiltonian cycle
in a planar drawing of~$G$,
and, hence, can all be placed onto Page~2.
\end{proof}

\section{Future Work}
We have introduced the notion of weighted book thickness.
Since this is a new concept, there are many questions and directions that allow for future research.
We list some of them:
\begin{itemize}
    \item What is the tight ratio between $\wbt_2$ and $\wbt_3$ for the  considered graph classes?
    \item Can the difference between $\wbt$ and $\wbt_{\bt}$ be arbitrarily large?
    \item Can the difference between $\bt$ and the min.~$k$ such that $\wbt=\wbt_k$ be arbitrarily large?
   \item Study more graph classes, for example, planar graphs or planar 3-trees.
    \item We have used a simple linear weight function,
    where an edge on page~$i$ has weight~$i$.
    This can be generalized to any (monotonically increasing) function.
    What about quadratically or exponentially growing weight functions?
    Can we make some general statements?
    \item Extend weighted book thickness to other linear layout types, e.g., queue layouts.
    \item Study other notions of graph thickness, e.g., geometric thickness, in the weighted setting.
\end{itemize}

\bibliographystyle{plainurl}
\bibliography{stacks,bibliography}
\end{document}